\documentclass[final,3p,times]{elsarticle}

\usepackage{amssymb}
\usepackage{amsmath}

\usepackage{caption}
\usepackage{amsmath,amssymb}
\usepackage{ntheorem}
\usepackage{amssymb}
\usepackage{hyperref}
\usepackage{tikz}
\usepackage{xcolor}
\usepackage{hyperref}
\usepackage{subfigure}
\usepackage{graphicx}
\usepackage{setspace}
\usepackage{soul, color, xcolor}    
\usepackage{epstopdf}
\usepackage{multirow}  
\usepackage{amsmath}   
\usepackage{graphicx}  
\usepackage{array}     
\usepackage{booktabs}  
\usepackage[ruled,vlined]{algorithm2e}    
\usepackage{algorithmic}
\usepackage{microtype}

\usepackage{multirow}
\usepackage{diagbox} 
\usepackage{booktabs} 

\makeatletter

\newcommand{\Rmnum}[1]{\expandafter\@slowromancap\romannumeral #1@}
\makeatother

\newtheorem{theorem}{Theorem}
\newtheorem{definition}{Definition}

\newtheorem{proof}{Proof}

\journal{Signal Processing}

\begin{document}

\begin{frontmatter}



\title{	A frame-theoretic two-dimensional multi-window graph fractional Fourier transform for product graph signal analysis
}


\author[1]{Linbo~Shang\corref{cor1}}

\ead{ashanglinbo@163.com}
\cortext[cor1]{Corresponding author.}
\address[1]{School of Mathematics and Statistics, Nanjing University of Information Science and Technology, Nanjing 210044, China}

\tnotetext[mytitlenote]{This work was supported by the Postgraduate Research \& Practice Innovation Program of Jiangsu Province under Grant KYCX25\_1559.}

\begin{abstract}
The analysis of multi-dimensional graph signals on complex structured domains remains a fundamental challenge, as existing windowed graph fractional Fourier transform methods often fail to effectively capture the intrinsic product structures and underlying characteristics of such data. Moreover, single-window approaches lack the flexibility required for accurate localized analysis.
To overcome these limitations, this paper proposes a novel framework termed the two-dimensional multi-window graph fractional Fourier transform (2D-MWGFRFT), specifically designed for signals defined on two-dimensional Cartesian product graphs. Within this framework, we introduce a class of two-dimensional multi-window graph fractional Fourier atoms and establish a rigorous frame-theoretic foundation, ensuring numerical stability and enabling perfect reconstruction. Furthermore, to enhance computational efficiency, we develop a fast implementation that exploits the separable structure of the underlying product graph. 
Experimental results demonstrate that the proposed 2D-MWGFRFT significantly improves representation capability and analytical flexibility compared to conventional single-window and one-dimensional approaches, thereby providing a powerful and robust tool for localized fractional vertex–frequency analysis of multi-dimensional graph signals.
\end{abstract}


\end{frontmatter}

\section{Introduction}\label{section1}

In modern data-driven applications, a large volume of real-world data naturally resides on irregular and interconnected domains that can be effectively modeled as graphs. Representative examples include social networks, sensor networks, transportation systems, biological interaction networks, and brain connectivity structures \cite{ref1,ref2,ref3,ref4,ref5,ref6}. To process such non-Euclidean data, Graph Signal Processing (GSP) \cite{ref7,ref8,ref9} has emerged as a powerful framework that generalizes classical signal processing tools---such as filtering, sampling, and transforms---to graph-structured domains \cite{ref10,ref11,ref12,ref13,ref14,ref15}. By leveraging spectral graph theory, GSP enables efficient analysis and interpretation of signals defined over complex and irregular structures.

Among the core tools in GSP, the Graph Fourier Transform (GFT) plays a fundamental role by generalizing the notion of frequency \cite{ref11,ref16} to graph domains via the eigendecomposition of the graph Laplacian. Building upon this foundation, fractional extensions \cite{ref17,ref18} such as the Graph Fractional Fourier Transform (GFRFT) and the Spectral Graph Fractional Fourier Transform (SGFRFT) have been developed to provide additional flexibility for signal analysis \cite{ref19,ref20}. These transforms enable a continuous transition between the vertex domain and the spectral domain, thereby offering enhanced capability for multi-scale and intermediate-domain signal representation.

In classical signal processing, windowed transforms have proven highly effective for capturing localized signal characteristics via joint time--frequency analysis \cite{ref25,ref26,ref27,ref28,ref29,ref30}. Inspired by this idea, windowed extensions have been developed in the graph setting \cite{ref31,ref33,ref34,Yan2021dsp}, among which the windowed graph fractional Fourier transform (WGFRFT) \cite{Yan2021dsp} enables localized vertex--frequency analysis. However, the use of a single window inherently imposes a trade-off between spatial and spectral resolution, limiting its adaptability to signals with complex and heterogeneous structures. To alleviate this limitation, multi-window strategies have recently attracted increasing attention \cite{zheng2016multi,Zheng2021cc,bulai2026ECarXiv}, leading to the development of the multi-window graph fractional Fourier transform (MWGFRFT) \cite{SHANG2026110198}, which improves flexibility and representation capability.

Despite these advances, most existing graph fractional transforms are primarily designed for one-dimensional graph signals and are not well suited for multi-dimensional graph data. In particular, for signals defined on Cartesian product graphs \cite{ref21,ref22,ref23,ref24}, conventional approaches typically collapse multi-dimensional spectral representations into a single frequency axis. Such simplifications inevitably lead to the loss of intrinsic structural and directional information, thereby limiting their descriptive power. Although several multidimensional extensions have been proposed, such as the two-dimensional spectral graph fractional Fourier transform (2D-SGFRFT) \cite{ref35} and the two-dimensional windowed graph fractional Fourier transform (2D-WGFRFT) \cite{GAN2025105191}, their capability to capture localized vertex--frequency characteristics remains insufficient.

To address the above challenges, this paper proposes a novel two-dimensional multi-window graph fractional Fourier transform (2D-MWGFRFT) defined on Cartesian product graphs. As illustrated in Fig.~\ref{2dFig1}, the proposed framework unifies and extends existing methods, including WGFRFT \cite{Yan2021dsp}, MWGFRFT \cite{SHANG2026110198}, and 2D-WGFRFT \cite{GAN2025105191}. By incorporating multi-window mechanisms into the fractional graph spectral framework, the proposed transform enables flexible multi-resolution analysis and significantly enhances vertex--frequency localization. Furthermore, we establish the associated two-dimensional multi-window graph fractional Fourier frame, which provides a rigorous theoretical foundation for stable signal representation and perfect reconstruction.
\begin{figure}[htbp]
	\centering
	\includegraphics[width=0.7\columnwidth]{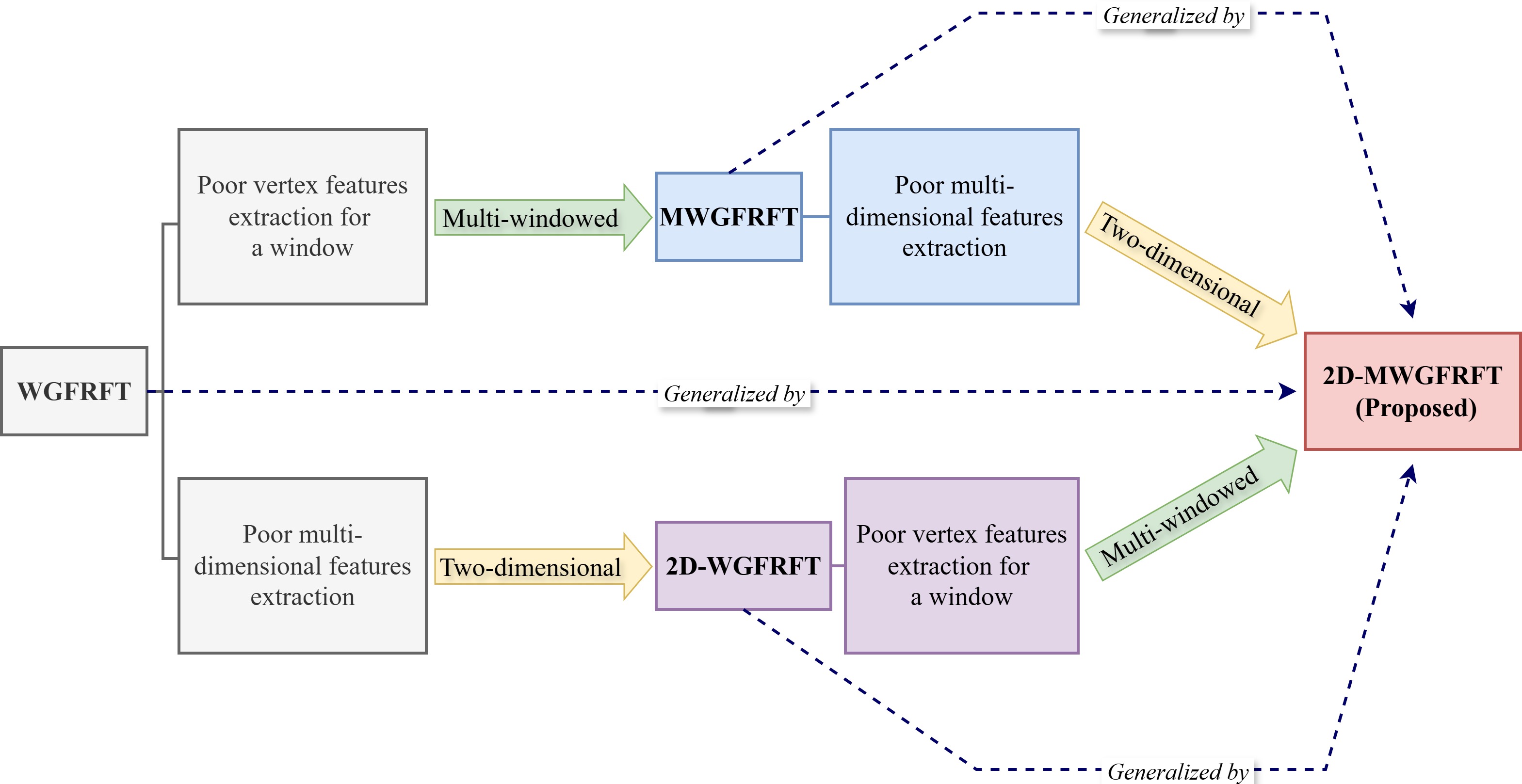}
	\caption{Relationship between WGFRFT, MWGFRFT, 2D-WGFRFT and 2D-MWGFRFT.}
	\label{2dFig1}
\end{figure}

To address computational challenges, we further develop a fast implementation, referred to as F2D-MWGFRFT. By exploiting the separable structure of Cartesian product graphs and reformulating the transform in the spectral domain, the proposed algorithm achieves a substantial reduction in computational complexity, making it suitable for large-scale graph signal processing applications.

The main contributions of this work are summarized as follows:
\begin{itemize}
	\item We propose a novel 2D-MWGFRFT framework that integrates multi-window analysis with graph fractional transforms, enabling enhanced vertex-frequency representation for signals on Cartesian product graphs.
	
	\item We construct the corresponding two-dimensional multi-window graph fractional Fourier frame and derive the forward and inverse transforms, providing a solid theoretical foundation for stable analysis and reconstruction.
	
	\item We develop an efficient fast algorithm (F2D-MWGFRFT), reducing the computational complexity from $O\big(L(N_1N_2)^4\big)$ to $O\big(L(N_1N_2)^3\big)$.
	
	\item We validate the effectiveness of the proposed method through extensive experiments, demonstrating superior performance in terms of spectral concentration, localization capability, and computational efficiency compared with existing approaches.
\end{itemize}

The remainder of this paper is organized as follows. Section~2 reviews the necessary preliminaries, including spectral graph theory, fractional transforms, and Cartesian product graphs. Section~3 introduces the proposed 2D-MWGFRFT and its fast implementation. Section~4 presents experimental results and performance evaluations. Section~5 demonstrates the application of the proposed method to anomaly detection. Finally, Section~6 concludes the paper.

\section{Preliminaries}\label{section2}

\subsection{Spectral graph theory}

Consider an undirected weighted graph $\mathcal{G}=\{\mathcal{V},\mathcal{E},W\}$, where $\mathcal{V}$ denotes a finite set of $N$ vertices, $\mathcal{E}$ represents the set of edges, and $W \in \mathbb{R}^{N \times N}$ is the weighted adjacency matrix. The non-normalized graph Laplacian, serving as a symmetric difference operator, is defined as $\mathcal{L}=D-W$ \cite{Shuman2013SPM}, where $D$ is the diagonal degree matrix with entries $D_{ii} = \sum_{j} W_{ij}$. The spectral decomposition of $\mathcal{L}$ is given by $\mathcal{L}=\chi\Lambda\chi^H$, where $\Lambda = \operatorname{diag}(\lambda_0, \lambda_1, \dots, \lambda_{N-1})$ represents the diagonal matrix of eigenvalues sorted as $0 = \lambda_0 < \lambda_1 \leq \dots \leq \lambda_{N-1}$, and $\chi = [\chi_0, \chi_1, \dots, \chi_{N-1}]$ is the unitary matrix containing the corresponding orthonormal eigenvectors.

For a graph signal $\mathbf{f} \in \mathbb{C}^N$, the graph Fourier transform (GFT) and its inverse (IGFT) are defined as:
\begin{equation*}
	\widehat{f}(\lambda_k) = \sum_{n=1}^{N}f(n)\chi_{k}^{*}(n), \quad f(n) = \sum_{k=0}^{N-1}\widehat{f}(\lambda_k)\chi_k(n).
\end{equation*}

To generalize this framework, the spectral graph fractional Fourier transform (SGFRFT) employs the fractional Laplacian operator $\mathcal{L}^{(\alpha)} = \gamma R \gamma^H$ for a fractional order $0 < \alpha \leq 1$. Here, $R = \Lambda^\alpha = \operatorname{diag}(r_0, r_1, \dots, r_{N-1})$ denotes the fractional eigenvalues, and $\gamma = \chi^\alpha$ represents the fractional unitary basis (often interpreted as a generalized basis derived from the fractional operator). The SGFRFT of a signal $\mathbf{f}$ and its inverse are expressed as:
\begin{equation*}
	\widehat{f}^{\alpha}(r_k) = \sum_{n=1}^{N}f(n)\gamma_{k}^{*}(n), \quad f(n) = \sum_{k=0}^{N-1}\widehat{f}^\alpha(r_k)\gamma_k(n).
\end{equation*}
Note that the SGFRFT reduces to the standard GFT when $\alpha=1$. Furthermore, the SGFRFT preserves the inner product, satisfying the Parseval relation: $\langle \mathbf{f}, \mathbf{g} \rangle = \langle \widehat{\mathbf{f}}^{\alpha}, \widehat{\mathbf{g}}^{\alpha} \rangle$.

\subsection{Windowed graph fractional Fourier transform and multi-window graph fractional Fourier transform}

\begin{definition}(Graph fractional translation operator)
	For any signal $\mathbf{g}$ defined on graph $\mathcal{G}$, the graph fractional translation operator $T_i^{\alpha}$ at vertex $i \in \{1, 2, \dots, N\}$ is defined as:
	\begin{equation}\label{eq1}
		(T_i^\alpha \mathbf{g})(n) = N^{\alpha/2} \sum_{p=0}^{N-1} \widehat{\mathbf{g}}^{\alpha}(r_p) \gamma_{p}^*(i) \gamma_{p}(n).
	\end{equation}
\end{definition}

\begin{definition}(Graph fractional modulation operator)
	For any signal $\mathbf{g}$ defined on graph $\mathcal{G}$, the graph fractional modulation operator at frequency index $k \in \{0, 1, \dots, N-1\}$ is defined as:
	\begin{equation*}
		(M_k^\alpha \mathbf{g})(n) = N^{\alpha/2} \mathbf{g}(n) \gamma_k(n).
	\end{equation*}
\end{definition}

Let $\{\mathbf{g}_l\}_{l=1}^L \subset \mathbb{C}^N$ be a finite sequence of $L$ window functions, where $\widehat{\mathbf{g}_l}^\alpha$ denotes the SGFRFT of $\mathbf{g}_l$. We construct a set of multi-window graph fractional Fourier atoms as:
\begin{equation}\label{eqma1}
	\mathcal{G}^{l\alpha} = \{ \mathbf{g}_{i,k}^{(l\alpha)} \mid i=1, \dots, N; \, k=0, \dots, N-1; \, l=1, \dots, L; \, 0 < \alpha \leq 1 \},
\end{equation}
where each atom is generated by the successive application of the modulation and translation operators:
\begin{equation*}
	\mathbf{g}_{i,k}^{(l\alpha)}(n) = (M_k^\alpha T_i^\alpha \mathbf{g}_l)(n) = N^{\alpha} \gamma_{k}(n) \sum_{p=0}^{N-1} \widehat{\mathbf{g}_l}^\alpha(r_p) \gamma_{p}^*(i) \gamma_{p}(n).
\end{equation*}

\begin{definition}\label{mwgfft1}
	(Multi-window graph fractional Fourier transform)
	Given a set of window functions $\{\mathbf{g}_l\}_{l=1}^L$, the multi-window graph fractional Fourier transform (MWGFRFT) of a signal $\mathbf{f} \in \mathbb{C}^N$, denoted by $Wf(\mathbf{g}_{i,k}^{(l\alpha)}) := \langle \mathbf{f}, \mathbf{g}_{i,k}^{(l\alpha)} \rangle$, is defined as:
	\begin{equation}\label{eqMW1}
		Wf(\mathbf{g}_{i,k}^{(l\alpha)}) = N^{\alpha} \sum_{l=1}^{L}\sum_{n=1}^{N} f(n) \left( \sum_{p=0}^{N-1} \widehat{\mathbf{g}_l}^{\alpha*}(r_p) \gamma_{p}(i) \gamma_{p}^{*}(n) \right) \gamma_{k}^{*}(n).
	\end{equation}
\end{definition}

It is worth noting that for $L=1$, Eq.~\eqref{eqMW1} simplifies to the standard windowed graph fractional Fourier transform (WGFRFT) as discussed in \cite{Yan2021dsp}.

\subsection{Cartesian product graph}\label{203}

Consider two weighted graphs $\mathcal{G}_1 = \{\mathcal{V}_1, \mathcal{E}_1, {W}_1\}$ and $\mathcal{G}_2 = \{\mathcal{V}_2, \mathcal{E}_2, {W}_2\}$, with Laplacian matrices denoted by $\mathcal{L}_1$ and $\mathcal{L}_2$, respectively. The vertex sets are defined as $\mathcal{V}_1 = \{1, \dots, N_1\}$ and $\mathcal{V}_2 = \{1, \dots, N_2\}$. 
The Cartesian product graph, denoted by $\mathcal{G} = \mathcal{G}_1 \square \mathcal{G}_2$, is constructed over the product vertex set $\mathcal{V}_1 \times \mathcal{V}_2$. Two vertices $(i_1, i_2)$ and $(j_1, j_2)$ in $\mathcal{G}$ are adjacent if and only if they satisfy one of the following conditions:
\begin{equation*}
	[\{i_1, j_1\} \in \mathcal{E}_1 \text{ and } i_2 = j_2] \quad \text{or} \quad [i_1 = j_1 \text{ and } \{i_2, j_2\} \in \mathcal{E}_2].
\end{equation*}
Accordingly, the weight function ${W}_\square$ for the product graph is defined as:
\begin{equation*}
	W_\square((i_1, i_2), (j_1, j_2)) = W_1(i_1, j_1)\delta(i_2, j_2) + \delta(i_1, j_1)W_2(i_2, j_2),
\end{equation*}
where $\delta(i, j)$ represents the Kronecker delta function ($\delta=1$ if $i=j$, and $0$ otherwise). Based on this construction, the Cartesian product operation satisfies the following algebraic properties:
\begin{itemize}
	\item \textbf{Associativity:} $(\mathcal{F} \square \mathcal{G}) \square \mathcal{H} \cong \mathcal{F} \square (\mathcal{G} \square \mathcal{H})$.
	\item \textbf{Commutativity:} $\mathcal{G} \square \mathcal{H} \cong \mathcal{H} \square \mathcal{G}$.
\end{itemize}

The graphs $\mathcal{G}_1$ and $\mathcal{G}_2$ are termed the factor graphs of $\mathcal{G}_1 \square \mathcal{G}_2$. The adjacency matrix ${W}$ and the Laplacian matrix $\mathcal{L}$ of the resulting product graph can be expressed using the Kronecker sum ($\oplus$) \cite{ref36} and Kronecker product ($\otimes$) of the matrices of its factor graphs:
\begin{equation}
	{W} = {W}_1 \oplus {W}_2 = {W}_1 \otimes \mathbf{I}_{N_2} + \mathbf{I}_{N_1} \otimes {W}_2,
\end{equation}
\begin{equation}
	\mathcal{L} = \mathcal{L}_1 \oplus \mathcal{L}_2 = \mathcal{L}_1 \otimes \mathbf{I}_{N_2} + \mathbf{I}_{N_1} \otimes \mathcal{L}_2,
\end{equation}
where $\mathbf{I}_n$ denotes an $n \times n$ identity matrix. For any matrices ${A} \in \mathbb{R}^{m \times n}$ and ${B} \in \mathbb{R}^{p \times q}$, the Kronecker product ${A} \otimes {B}$ yields a $pm \times qn$ block matrix:
\begin{equation}
	{A} \otimes {B} = \begin{pmatrix} a_{11}{B} & \cdots & a_{1n}{B} \\ \vdots & \ddots & \vdots \\ a_{m1}{B} & \cdots & a_{mn}{B} \end{pmatrix}.
\end{equation}

\subsection{Two-dimensional spectral graph fractional Fourier transform}

In light of the Cartesian product graph $\mathcal{G}_1 \square \mathcal{G}_2$ introduced in Section \ref{203}, let the Laplacian operators of the factor graphs be defined by their eigendecompositions: $\mathcal{L}_1 = \chi_1 \Lambda_1 \chi_1^H$ and $\mathcal{L}_2 = \chi_2 \Lambda_2 \chi_2^H$. The corresponding graph fractional Laplacian operators for $\mathcal{G}_1$ and $\mathcal{G}_2$ are then defined as:
\begin{equation*}
\mathcal{L}_j^{(\alpha)} = \gamma^{(j)} R^{(j)} (\gamma^{(j)})^H,
\end{equation*}
where$\gamma^{(j)} = [\gamma_0^{(j)}, \gamma_1^{(j)}, \dots, \gamma_{N_j-1}^{(j)}] = \chi_j^{\alpha}$ denotes the fractional eigenvector matrix and $R^{(j)} = \mathrm{diag}([r_0^{(j)}, r_1^{(j)}, \dots, r_{N_j-1}^{(j)}]) = \Lambda_j^{\alpha}$ represents the fractional eigenvalue matrix for $j=1,2$.

For the Cartesian product graph $\mathcal{G}_1 \square \mathcal{G}_2$, the eigendecomposition of the Kronecker sum $\mathcal{L}_1^{(\alpha)} \oplus \mathcal{L}_2^{(\alpha)}$ in the graph fractional domain is given by:
\begin{equation*}
	(\mathcal{L}_1^{(\alpha)} \oplus \mathcal{L}_2^{(\alpha)} )\left( \gamma_{k_{1}}^{(1)} \otimes \gamma_{k_{2}}^{(2)} \right) = (r_{k_{1}}^{(1)} + r_{k_{2}}^{(2)}) \left( \gamma_{k_{1}}^{(1)} \otimes \gamma_{k_{2}}^{(2)} \right),
\end{equation*}
where the joint eigenvector is the Kronecker product of the individual fractional eigenvectors, with its elements expressed as $\gamma_{k_{1}}^{(1)}(n_1) \gamma_{k_{2}}^{(2)}(n_2)$.

The two-dimensional spectral graph fractional Fourier transform (2D-SGFRFT) of a two-dimensional signal $\mathbf{f}: \mathcal{V}_1 \times \mathcal{V}_2 \to \mathbb{C}$ on $\mathcal{G}_1 \square \mathcal{G}_2$ is defined as:
\begin{equation*}
	\widehat{f}^{\alpha}(r_{k_{1}}, r_{k_{2}}) = \sum_{n_{1}=1}^{N_{1}} \sum_{n_{2}=1}^{N_{2}} f(n_{1}, n_{2}) (\gamma_{k_{1}}^{(1)}(n_{1}))^* (\gamma_{k_{2}}^{(2)}(n_{2}))^*,
\end{equation*}
for $k_{1} = 0, \dots, N_{1}-1$ and $k_{2} = 0, \dots, N_{2}-1$. The corresponding inverse transform is given by:
\begin{equation*}
	f(n_1, n_2) = \sum_{k_1=0}^{N_1-1} \sum_{k_2=0}^{N_2-1} \widehat{f}^{\alpha}(r_{k_{1}}, r_{k_{2}}) \gamma_{k_1}^{(1)}(n_1) \gamma_{k_2}^{(2)}(n_2).
\end{equation*}

\subsection{Two-dimensional windowed graph fractional Fourier transform}

The graph fractional Laplacian operator of $\mathcal{G}_{\square}=\mathcal{G}_1 \square \mathcal{G}_2$ is formulated as
$$\mathcal{L}^{(\alpha)} = \mathcal{L}_1^{(\alpha)} \oplus \mathcal{L}_2^{(\alpha)} = \gamma^{prod} R^{prod}(\gamma^{prod})^H.$$
which satisfies the eigenvalue equation:
$$
\mathcal{L}^{(\alpha)} 
\begin{pmatrix} \gamma_{k_1}^{(1)}(1) \gamma_{k_2}^{(2)}(1) \\ \gamma_{k_1}^{(1)}(1) \gamma_{k_2}^{(2)}(2) \\ \vdots \\ \gamma_{k_1}^{(1)}(N_1) \gamma_{k_2}^{(2)}(N_2) \end{pmatrix}
= (r_{k_1}^{(1)} + r_{k_2}^{(2)})
\begin{pmatrix} \gamma_{k_1}^{(1)}(1) \gamma_{k_2}^{(2)}(1) \\ \gamma_{k_1}^{(1)}(1) \gamma_{k_2}^{(2)}(2) \\ \vdots \\ \gamma_{k_1}^{(1)}(N_1) \gamma_{k_2}^{(2)}(N_2) \end{pmatrix}.
$$
For any two-dimensional signals ${f}, {g}: \mathcal{V}_1 \times \mathcal{V}_2 \to \mathbb{C}$ on $\mathcal{G}_1 \square \mathcal{G}_2$, the two-dimensional graph fractional convolution operator $*_\alpha$ is defined as:
$$({f} *_\alpha {g})(n_1, n_2) = \sum_{k_1=0}^{N_1-1} \sum_{k_2=0}^{N_2-1} \widehat{f}^\alpha(k_1, k_2) \widehat{g}^\alpha(k_1, k_2) \gamma_{k_1}^{(1)}(n_1) \gamma_{k_2}^{(2)}(n_2).$$

Equipped with this convolution, for any signal ${f}\in\mathbb{C}^{N_1 \times N_2}$ and any vertex $(i_1, i_2)$, we define the two-dimensional graph fractional translation operator $T_{i_1,i_2}^{\alpha}:\mathbb{C}^{N_1\times N_2}\to\mathbb{C}^{N_1\times N_2}$ via a generalized convolution with a delta function $\delta_{i_1,i_2}$ centered at vertex $(n_1, n_2)$:
$$(T_{i_1,i_2}^{\alpha}f)\left(n_1,n_2\right)
:=(N_1N_2)^{\alpha/2}(f*_\alpha\delta_{i_1,i_2})(n_1,n_2).$$

\begin{definition}\label{def2dmwt2}
	(Two-dimensional graph fractional translation operator) 
	For any signal $\mathbf{g}\in\mathbb{C}^{N_1 \times N_2}$ and vertex $(i_1, i_2)$ on $\mathcal{G}_1 \square \mathcal{G}_2$, the two-dimensional graph fractional translation operator $T_{i_1,i_2}^\alpha$ is formally expressed as:
	\begin{equation}
		(T_{i_1,i_2}^\alpha \mathbf{g})(n_1,n_2) = (N_1N_2)^{\alpha/2} \sum_{p_1=0}^{N_1-1}\sum_{p_2=0}^{N_2-1} \widehat{\mathbf{g}}^{\alpha}(r_{p_1}^{(1)}, r_{p_2}^{(2)}) (\gamma_{p_1}^{(1)}(i_1) \gamma_{p_2}^{(2)}(i_2))^* \gamma_{p_1}^{(1)}(n_1)\gamma_{p_2}^{(2)}(n_2).
	\end{equation}
\end{definition}

\begin{definition}
	(Two-dimensional graph fractional modulation operator)
	For a signal $\mathbf{g}\in\mathbb{C}^{N_1 \times N_2}$ on $\mathcal{G}_1 \square \mathcal{G}_2$ and any fractional frequency indices $k_1 \in \{0, \dots, N_1-1\}$, $k_2 \in \{0, \dots, N_2-1\}$, the two-dimensional graph fractional modulation operator is defined as:
	\begin{equation*}
		(M_{k_1,k_2}^\alpha \mathbf{g})(n_1,n_2) = (N_1N_2)^{\alpha/2} \mathbf{g}(n_1,n_2)\gamma_{k_1}^{(1)}(n_1)\gamma_{k_2}^{(2)}(n_2).
	\end{equation*}
\end{definition}

\begin{definition}
	(Two-dimensional windowed graph fractional Fourier transform)
	Given a window function $\mathbf{g}: \mathcal{V}_1 \times \mathcal{V}_2 \to \mathbb{C}$ and its 2D-SGFRFT $\widehat{\mathbf{g}}^{\alpha}$, the two-dimensional windowed graph fractional Fourier transform (2D-WGFRFT) of a signal $\mathbf{f}$ on $\mathcal{G}_1 \square \mathcal{G}_2$ is defined by:
	\begin{equation*}
		\begin{aligned}
			W\mathbf{f}(i_1,i_2,k_1,k_2) = & (N_1N_2)^\alpha \sum_{n_1=1}^{N_1}\sum_{n_2=1}^{N_2}\mathbf{f}(n_1,n_2) (\gamma_{k_1}^{(1)}(n_1)\gamma_{k_2}^{(2)}(n_2))^* \\
			& \times \left(\sum_{p_1=0}^{N_1-1}\sum_{p_2=0}^{N_2-1} \widehat{\mathbf{g}}^{\alpha*}(r_{p_1}^{(1)}, r_{p_2}^{(2)}) \gamma_{p_1}^{(1)}(i_1)\gamma_{p_2}^{(2)}(i_2) (\gamma_{p_1}^{(1)}(n_1)\gamma_{p_2}^{(2)}(n_2))^* \right).
		\end{aligned}
	\end{equation*}
\end{definition}

\section{The two-dimensional multi-window graph fractional Fourier transform}
\label{section3}

In this section, we formulate the two-dimensional multi-window graph fractional Fourier frames and establish the associated theoretical framework.

\subsection{Two-dimensional multi-window graph fractional Fourier transform}

Consider a finite sequence of window functions $g_1, g_2, \ldots, g_L : \mathcal{V}_1 \times \mathcal{V}_2 \to \mathbb{C}$, where the 2D-SGFRFT of each $g_l$ is denoted by $\widehat{g}_l^\alpha$. We define the set of two-dimensional multi-windowed graph fractional Fourier atoms as:
\begin{equation}\label{eq2dmw03}
	\mathcal{G}^{l\alpha} = \big\{
	g_{i_1,i_2;k_1,k_2}^{(l\alpha)}
	\mid
	i_1=1,\ldots,N_1; i_2=1,\ldots,N_2;
	k_1=0,\ldots,N_1-1; k_2=0,\ldots,N_2-1; l=1,\ldots,L; 0<\alpha \leq 1
	\big\},\end{equation}
where the atoms are constructed as:
\begin{equation*}
	\begin{aligned}
		g_{i_1,i_2;k_1,k_2}^{(l\alpha)}(n_1,n_2) 
		&=(M_{k_1,k_2}^\alpha T_{i_1,i_2}^\alpha g_l)(n_1,n_2) \\
		&=(N_1N_2)^\alpha
		\gamma_{k_1}^{(1)}(n_1)
		\gamma_{k_2}^{(2)}(n_2)
		\left(\sum\limits_{p_1=0}^{N_1-1}\sum\limits_{p_2=0}^{N_2-1}
		\widehat{g}_l^{\alpha}\left(r_{p_1}^{(1)}+r_{p_2}^{(2)}\right)
		\overline{\gamma_{p_1}^{(1)}(i_1)\gamma_{p_2}^{(2)}(i_2)}
		\gamma_{p_1}^{(1)}(n_1)\gamma_{p_2}^{(2)}(n_2) \right).
	\end{aligned}
\end{equation*}

\begin{theorem}\label{th2d02}
	Let $\mathcal{G}^{l\alpha}$ be the set of two-dimensional multi-window graph fractional Fourier atoms defined in \eqref{eq2dmw03}. If $\sum\limits_{l=1}^L |\widehat{g}_l^\alpha(0,0)|^2 \neq 0$, then $\mathcal{G}^{l\alpha}$ constitutes a frame (i.e., a 2D-MWGFRFF). Consequently, for any signal $f \in \mathbb{C}^{N_1 \times N_2}$, the following frame condition holds:
	\begin{equation*}
		0 < A \|f\|_2^2 
		\leq \sum_{l=1}^L 
		\sum_{i_1=1}^{N_1} \sum_{i_2=1}^{N_2}
		\sum_{k_1=0}^{N_1-1} \sum_{k_2=0}^{N_2-1} 
		|\langle f, g_{i_1,i_2;k_1,k_2}^{(l\alpha)} \rangle|^2 
		\leq B \|f\|_2^2 
		< \infty,
	\end{equation*}
	where the lower frame bound $A$ and upper frame bound $B$ are given by:
	\begin{equation}\label{eq2dmw05}
		0 < A = \min_{\substack{n_1 \in \{1,\ldots,N_1\}\\n_2 \in \{1,\ldots,N_2\}}} 
		\left\{(N_{1}N_{2})^{\alpha} \sum_{l=1}^L \|T_{n_1,n_2}^{\alpha} g_l\|_2^2\right\},
	\end{equation}
	\begin{equation}\label{eq2dmw06}
		B = \max_{\substack{n_1 \in \{1,\ldots,N_1\}\\n_2 \in \{1,\ldots,N_2\}}} 
		\left\{(N_{1}N_{2})^{\alpha} \sum_{l=1}^L \|T_{n_1,n_2}^{\alpha} g_l\|_2^2\right\} < \infty.
	\end{equation}
\end{theorem}

\begin{proof}
	See \ref{2datheorem01}.
\end{proof}

\begin{definition}\label{def2dmw1}
	(Two-dimensional multi-window graph fractional Fourier transform)
	Given a finite sequence of window functions $\{g_l\}_{l=1}^L : \mathcal{V}_1 \times \mathcal{V}_2 \to \mathbb{C}$, the 2D-SGFRFT of $g_l$ is $\widehat{g}_l^\alpha$.
	The two-dimensional multi-window graph fractional Fourier transform
	$Wf(g_{i_1,i_2;k_1,k_2}^{(l\alpha)})
	:=\langle{f},g_{i_1,i_2;k_1,k_2}^{(l\alpha)}(n_1,n_2)\rangle$
	of two-dimensional signal ${f}$ on graph $\mathcal{G}_1 \square \mathcal{G}_2$ is defined by
	\begin{equation}
		\begin{aligned}\label{eq2dmw01}
			Wf(g_{i_1,i_2;k_1,k_2}^{(l\alpha)})
			=(N_1N_2)^\alpha
			\sum_{l=1}^{L}
			\sum_{n_1=1}^{N_1}\sum_{n_2=1}^{N_2}f(n_1,n_2)
			\overline{\gamma_{k_1}^{(1)}(n_1)\gamma_{k_2}^{(2)}(n_2)}
			\left(\sum_{p_1=0}^{N_1-1}\sum_{p_2=0}^{N_2-1}
			\widehat{\mathbf{g}}_{l}^{\alpha*}\left(r_{p_1}^{(1)}+r_{p_2}^{(2)}\right)
			\gamma_{p_1}^{(1)}(i_1)\gamma_{p_2}^{(2)}(i_2)
			\overline{\gamma_{p_1}^{(1)}(n_1)\gamma_{p_2}^{(2)}(n_2)} \right).
		\end{aligned}
	\end{equation}
\end{definition}
When $L = 1$, Eq.\eqref{eq2dmw01} degenerates into standard 2D-WGFRFT.

\begin{theorem}\label{th2d01}
	 If the window functions $\{g_l\}_{l=1}^L: \mathcal{V}_1 \times \mathcal{V}_2 \to \mathbb{C}$ satisfies $T_{i_1, i_2}^{\alpha}(g_l) \neq 0$, for any two-dimensional signal ${f}$, the inverse two-dimensional multi-window graph fractional Fourier transform (I2D-MWGFRFT) is given by
	\begin{equation}\begin{aligned}\label{eq2dmw02}
		f(n_1,n_2)=\frac{1}{(N_{1}N_{2})^{\alpha} \sum\limits_{l=1}^{L}\left\|T_{i_{1},i_{2}}^{\alpha}g_l\right\|_{2}^{2}}
		\times\sum_{i_{1}=1}^{N_{1}}\sum_{i_{2}=1}^{N_{2}}
		\sum_{k_{1}=0}^{N_{1}-1}\sum_{k_{2}=0}^{N_{2}-1}
	Wf(g_{i_1,i_2;k_1,k_2}^{(l\alpha)})
g_{i_1,i_2;k_1,k_2}^{(l\alpha)}(n_1,n_2).
	\end{aligned}\end{equation}
\end{theorem}

\begin{proof}
	See \ref{2datheorem03}.
\end{proof}

\subsection{Fast algorithm}

Let us consider a Cartesian product graph $\mathcal{G}_1 \square \mathcal{G}_2$ comprising $N_1 \times N_2$ vertices. The direct computation of the two-dimensional multi-window graph fractional Fourier transform , denoted as $Wf(g_{i_1,i_2;k_1,k_2}^{(l\alpha)})$, involves $(N_1 \times N_2)^2$ elements. Since calculating each individual element entails a complexity of $O((N_1 \times N_2)^2)$, the total computational burden scales as $O(L(N_1 \times N_2)^4)$. To mitigate this prohibitive cost, we propose an efficient fast calculation method.

First, we define the $\Theta$-field via the two-dimensional spectral graph fractional Fourier transform of the 2D-MWGFRFT as follows:
\begin{equation*}
	\Theta(k_1, k_2, k_1', k_2') = \sum_{i_1=1}^{N_1} \sum_{i_2=1}^{N_2} Wf(g_{i_1,i_2;k_1^\prime,k_2^\prime}^{(l\alpha)}) \overline{\gamma_{k_1}^{(1)}(i_1) \gamma_{k_2}^{(2)}(i_2)}.
\end{equation*}
By substituting the definition of the transform, the $\Theta$-field can be expanded as:
\begin{equation*}
	\begin{aligned}
		&\Theta(k_1, k_2, k_1', k_2') \\
		=& \sum_{l=1}^{L} \sum_{i_1=1}^{N_1} \sum_{i_2=1}^{N_2} \left( (N_1N_2)^\alpha \sum_{n_1=1}^{N_1}\sum_{n_2=1}^{N_2}f(n_1,n_2) \overline{\gamma_{k_1^\prime}^{(1)}(n_1)\gamma_{k_2^\prime}^{(2)}(n_2)} \right. \\
		& \times \left. \left( \sum_{p_1=0}^{N_1-1}\sum_{p_2=0}^{N_2-1} \widehat{g}_{l}^{\alpha*}\left(r_{p_1},r_{p_2}\right) \gamma_{p_1}^{(1)}(i_1)\gamma_{p_2}^{(2)}(i_2) \overline{\gamma_{p_1}^{(1)}(n_1)\gamma_{p_2}^{(2)}(n_2)} \right) \overline{\gamma_{k_1}^{(1)}(i_1) \gamma_{k_2}^{(2)}(i_2)} \right) \\
		=& (N_1N_2)^\alpha \sum_{l=1}^{L} \sum_{p_1=0}^{N_1-1}\sum_{p_2=0}^{N_2-1} \sum_{n_1=1}^{N_1} \sum_{n_2=1}^{N_2} f(n_1,n_2) \overline{\gamma_{k_1^\prime}^{(1)}(n_1)\gamma_{k_2^\prime}^{(2)}(n_2)} \widehat{g}_{l}^{\alpha*}\left(r_{p_1},r_{p_2}\right) \overline{\gamma_{p_1}^{(1)}(n_1)\gamma_{p_2}^{(2)}(n_2)} \\
		& \times \left( \sum_{i_1=1}^{N_1}\sum_{i_2=1}^{N_2} \gamma_{p_1}^{(1)}(i_1)\gamma_{p_2}^{(2)}(i_2) \overline{\gamma_{k_1}^{(1)}(i_1) \gamma_{k_2}^{(2)}(i_2)} \right) \\
		=& (N_1N_2)^\alpha \sum_{l=1}^{L} \sum_{n_1=1}^{N_1} \sum_{n_2=1}^{N_2} f(n_1,n_2) \overline{\gamma_{k_1^\prime}^{(1)}(n_1)\gamma_{k_2^\prime}^{(2)}(n_2)} \overline{\gamma_{k_1}^{(1)}(n_1)\gamma_{k_2}^{(2)}(n_2)} \widehat{g}_{l}^{\alpha*}\left(r_{k_1},r_{k_2}\right) \\
		=& (N_1N_2)^\alpha \tilde{f}(k_{1}, k_{2}, k_{1}^{\prime}, k_{2}^{\prime}) \sum_{l=1}^{L} \widehat{g}_{l}^{\alpha*}(r_{k_1},r_{k_2}),
	\end{aligned}
\end{equation*}
where the auxiliary term $\tilde{f}$ is defined as:
\begin{equation*}
	\tilde{f}(k_{1}, k_{2}, k_{1}^{\prime}, k_{2}^{\prime}) = \sum_{n_1=1}^{N_1} \sum_{n_2=1}^{N_2} f(n_1,n_2) \overline{\gamma_{k_1^\prime}^{(1)}(n_1)\gamma_{k_2^\prime}^{(2)}(n_2)} \overline{\gamma_{k_1}^{(1)}(n_1)\gamma_{k_2}^{(2)}(n_2)}.
\end{equation*}
Building upon this, we formalize the efficient implementation of the 2D-MWGFRFT.

\begin{definition}\label{def2dmw2}
	(Fast two-dimensional multi-window graph fractional Fourier transform, i.e., F2D-MWGFRFT)
	The F2D-MWGFRFT is obtained by applying the inverse 2D-SGFRFT to the $\Theta$-field with respect to the variables $(k_1, k_2)$:
	\begin{equation*}
		\begin{aligned}
			FWf(g_{i_1,i_2;k_1^\prime,k_2^\prime}^{(l\alpha)}) 
			&= \sum_{k_1=0}^{N_1-1} \sum_{k_2=0}^{N_2-1} \Theta(k_1, k_2, k_1', k_2') \gamma_{k_1}^{(1)}(i_1) \gamma_{k_2}^{(2)}(i_2) \\
			&= (N_1N_2)^\alpha
			\sum_{l=1}^{L} \sum_{k_1=0}^{N_1-1} \sum_{k_2=0}^{N_2-1}  
			\tilde{f}(k_{1}, k_{2}, k_{1}^{\prime}, k_{2}^{\prime}) \widehat{g}_{l}^{\alpha*}(r_{k_1}^{(1)},r_{k_2}^{(2)}) \gamma_{k_1}^{(1)}(i_1) \gamma_{k_2}^{(2)}(i_2).
		\end{aligned}
	\end{equation*}
\end{definition}

The algorithmic procedure for the F2D-MWGFRFT is organized into the following steps:

\begin{enumerate}
	\item[\textbf{(i)}] \textbf{Spectral Window Computation:} Precompute $\widehat{g}_{l}^{\alpha*}$ for $l = 1, \ldots, L$. This step incurs a computational complexity of $O(L)$.
	
	\item[\textbf{(ii)}] \textbf{Auxiliary Matrix Construction:} Calculate $\tilde{f}(k_{1}, k_{2}, k_{1}^{\prime}, k_{2}^{\prime})$ using matrix operations:
	\begin{equation}
		\tilde{\mathbf{F}} = \left[ \tilde{f}(k_{1}, k_{2}, k_{1}^{\prime}, k_{2}^{\prime}) \right] = \left( \mathbf{F} \circ (\gamma_{k_1}^{(1)}\gamma_{k_2}^{(2)})^H \right)* 
		(\gamma_{k_1}^{(1)}\gamma_{k_2}^{(2)}),
	\end{equation}
	where $\circ$ represents Hadamard product and $*$ denotes the standard matrix multiplication. Each row of $\mathbf{F}$ contains the signal $f$, and $\left(\gamma_{k_1}^{(1)}\gamma_{k_2}^{(2)}\right)^H$ denotes the complex conjugate of the graph fractional basis $\gamma_{k_1}^{(1)}\gamma_{k_2}^{(2)}$. This step has $O\left((N_1 \times N_2)^3\right)$ computational complexity.

\item[\textbf{(iii)}] \textbf{Spectral Window Matrix Expansion:} 
To facilitate the subsequent operator mapping, the spectral window matrix $\mathbf{G}_l^{\alpha} \in \mathbb{C}^{N_1 \times N_2}$ is constructed via the 2D-SGFRFT of the window function $g_l$, defined as:
\begin{equation*}
	\mathbf{G}_l^{\alpha} = \begin{pmatrix}
		\hat{g}_l^{\alpha} (r_0^{(1)},r_0^{(2)}) & \hat{g}_l^{\alpha} (r_0^{(1)},r_1^{(2)}) & \cdots & \hat{g}_l^{\alpha} (r_0^{(1)},r_{N_2-1}^{(2)}) \\
		\hat{g}_l^{\alpha} (r_1^{(1)},r_0^{(2)}) & \hat{g}_l^{\alpha} (r_1^{(1)},r_1^{(2)}) & \cdots & \hat{g}_l^{\alpha} (r_1^{(1)},r_{N_2-1}^{(2)}) \\
		\vdots & \vdots & \ddots & \vdots \\
		\hat{g}_l^{\alpha} (r_{N_1-1}^{(1)},r_0^{(2)}) & \hat{g}_l^{\alpha} (r_{N_1-1}^{(1)},r_1^{(2)}) & \cdots & \hat{g}_l^{\alpha} (r_{N_1-1}^{(1)},r_{N_2-1}^{(2)})
	\end{pmatrix}.
\end{equation*}
The matrix $\mathbf{G}_l^{\alpha}$ is subsequently augmented into a higher-dimensional operator $\mathbf{\Psi}^{l\alpha} \in \mathbb{C}^{(N_1 N_2) \times (N_1 N_2)}$. 
Specifically, $\mathbf{G}_l^{\alpha}$ is first vectorized into a row vector 
$ [\text{vec}(\mathbf{G}_l^{\alpha})]^T \in \mathbb{C}^{1 \times (N_1 N_2)}$. 
This vector is then expanded into the matrix $\mathbf{\Psi}^{l\alpha}$ by replicating it across $(N_1 N_2)$ rows to align with the dimensions of the kernel matrix $\tilde{\mathbf{F}}$. This expansion is mathematically formulated as:
\begin{equation}
	\mathbf{\Psi}^{l\alpha} = \mathbf{1}_{(N_1 N_2) \times 1} \otimes [ \text{vec} ( \mathbf{G}_l^{\alpha} ) ]^T,
\end{equation}
where $\mathbf{1}_{(N_1 N_2) \times 1}$ denotes an all-ones column vector and $\otimes$ denotes the Kronecker product. This configuration ensures that each row of $\mathbf{\Psi}^{l\alpha}$ preserves the complete spectral profile of the window function, thereby enabling efficient element-wise operations. Notably, this formation step is primarily a structural reorganization and does not increase the overall computational complexity of the algorithm.

	\item[\textbf{(iv)}] \textbf{$\Theta$-Domain Transformation:} Compute the $\Theta$-domain representation as:
	\begin{equation}
		\Theta = (N_1N_2)^\alpha * (\Psi^{l\alpha})^* \circ \tilde{\mathbf{F}},
	\end{equation}
	which requires $O((N_1 \times N_2)^3)$ operations.
	
	\item[\textbf{(v)}] \textbf{Inverse Transform:} The vertex-frequency analysis is retrieved by applying the inverse 2D-SGFRFT to each row of $\Theta$:
	\begin{equation}
		FWf_l = (\gamma_{k_1}^{(1)}\gamma_{k_2}^{(2)})* \mathbf{\Theta}^{.T},
	\end{equation}
	where $\mathbf{\Theta}^{.T}$ denotes the non-conjugate transpose of $\mathbf{\Theta}$. The complexity of this step is $O((N_1 \times N_2)^3)$.
	
	\item[\textbf{(vi)}] \textbf{Multi-window Aggregation:} Perform matrix addition across the $L$ window-specific matrices. This step has a complexity of $O(L(N_1 \times N_2)^2)$.
\end{enumerate}

In summary, the total computational complexity of the F2D-MWGFRFT is $O(L(N_1 \times N_2)^3)$. As the dimensions $N_1$ and $N_2$ grow, this optimized approach provides a significant enhancement in efficiency compared to the original transformation's $O(L(N_1 \times N_2)^4)$ complexity.

\section{Experiments}\label{section4}

In this section, we conduct a series of numerical experiments to validate the theoretical framework of the two-dimensional multi-window graph fractional Fourier transform (2D-MWGFRFT) and its fast counterpart, the F2D-MWGFRFT.
Specifically, we first demonstrate the superiority of the 2-D transforms over their 1-D versions, highlighting their advantages in capturing structural and directional information on product graphs. 
Next, the algorithmic efficiency of the F2D-MWGFRFT is evaluated in terms of computational complexity and execution time. 
Subsequently, we analyze the impact of the parameter $L$ on the concentration and resolution of vertex-frequency representations. 
Finally, a comparative study between the F2D-WGFRFT and the F2D-MWGFRFT is performed to verify the effectiveness of the multi-window mechanism in vertex-frequency feature extraction and graph signal characterization.

\subsection{Superiority of 2-D Transforms over 1-D Counterparts}

The primary motivation for extending the windowed graph fractional Fourier transform to a 2-D formulation is to overcome the structural limitations inherent in 1-D graph signal processing when dealing with multi-dimensional relational data. To intuitively illustrate the significance and advantages of 2-D graph fractional operators, we conduct a comparative spectral analysis.

\subsubsection{Experimental configuration and signal synthesis}
We consider two factor path graphs, $\mathcal{G}_1$ and $\mathcal{G}_2$, with $N_1 = 9$ and $N_2 = 12$ vertices, respectively. 
The Cartesian product graph is defined as $\mathcal{G}_{\square} = \mathcal{G}_1 \square \mathcal{G}_2$, resulting in a composite network with $N_{total} = 108$ vertices. 
A Kronecker-structured joint fractional basis is constructed as $\gamma^{prod} = \gamma_{k_1}^{(1)} \otimes \gamma_{k_2}^{(2)}$, where $\gamma_{k_i}^{(i)}$ denotes the eigenvector of the fractional Laplacian operator $\mathcal{L}_i^{(\alpha)}$ for the $i$-th factor graph ($i=1,2$).

To simulate a non-stationary localized disturbance in the vertex domain, we synthesize a graph signal $f_1$ characterized by a localized impulse:
\begin{equation*}
	{f}_1(n) = 
	\begin{cases} 
		1, & n \in \{27, 28, 29\}; \\
		0, & \text{otherwise.}
	\end{cases}
\end{equation*}
This represents a localized excitation situated on the high-dimensional product topology. The spectral framework is implemented with a fractional order $\alpha = 0.7$ and a low-pass spectral window $g(\lambda) = \exp(-\tau \lambda)$ with a decay parameter $\tau=2$, serving as a localized spectral kernel for vertex-frequency feature extraction.

\subsubsection{Spectral characterization}

The experimental results—including the input signal, the corresponding spectral representations, and the filtered vertex-domain signal—are illustrated in Fig. \ref{2dfig101}.

\begin{figure*}[!t]
	\centering
	\subfigure[ a graph signal $f_1$ ]{\includegraphics[width=0.245\textwidth]{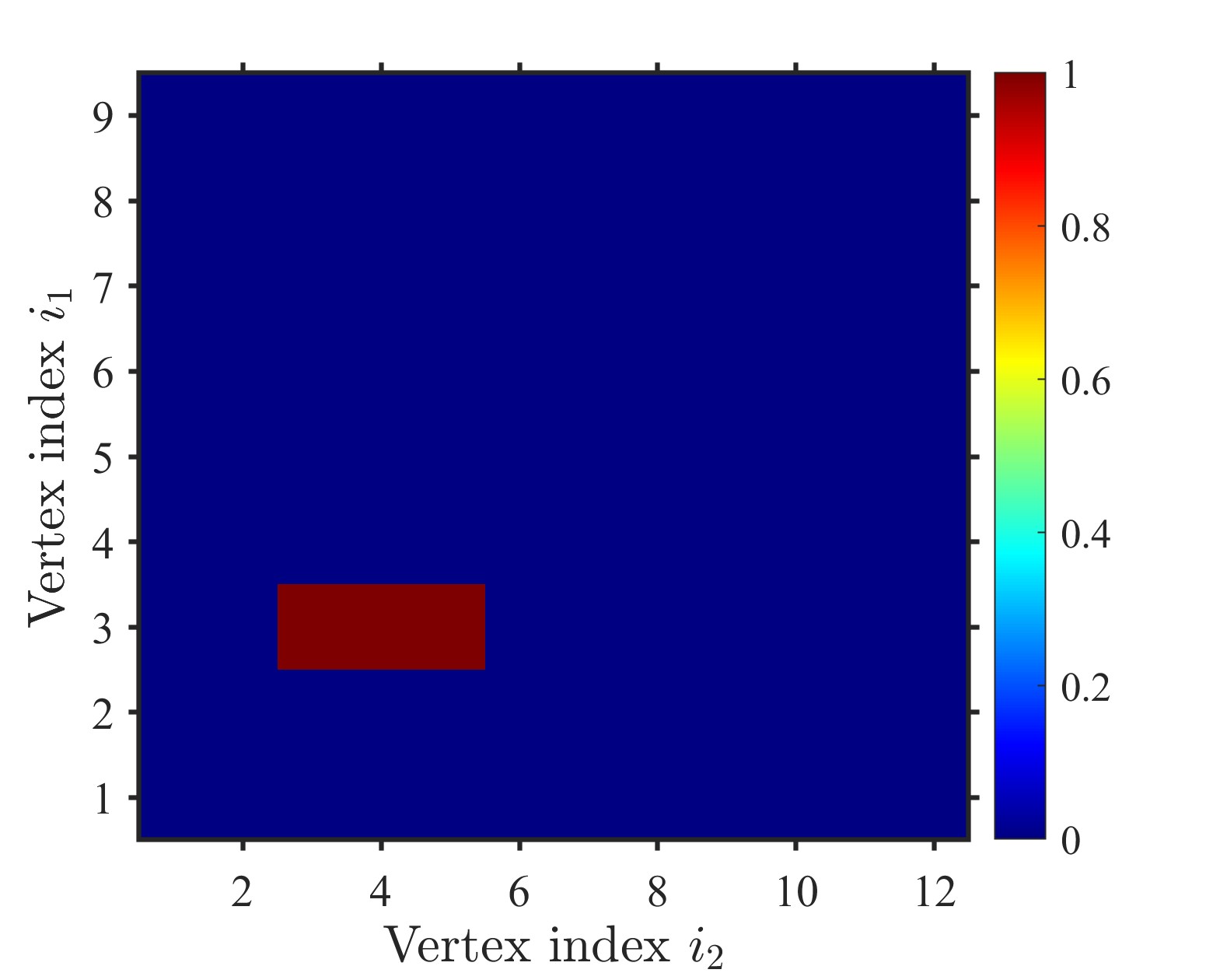}} \hfill
	\subfigure[ 1-D spectrum ]{\includegraphics[width=0.245\textwidth]{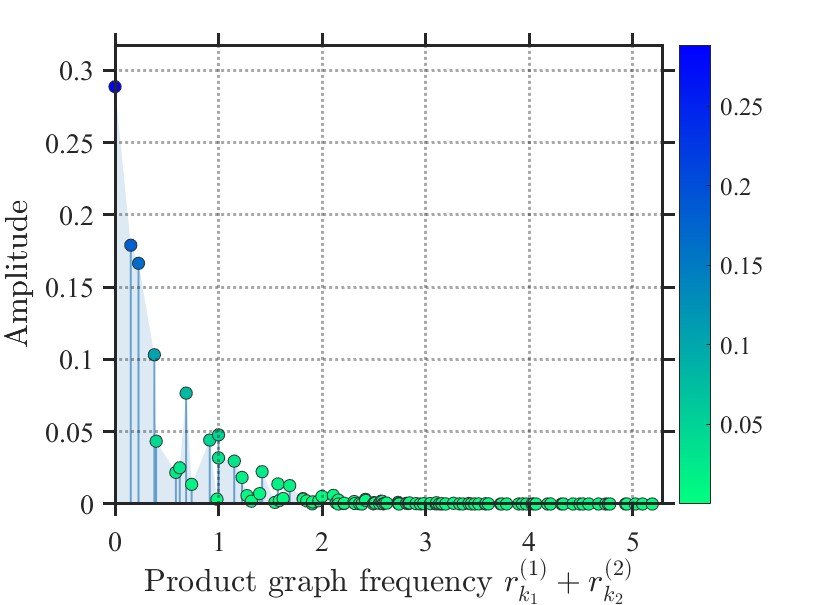}} \hfill
	\subfigure[2-D spectrum
	]{\includegraphics[width=0.245\textwidth]{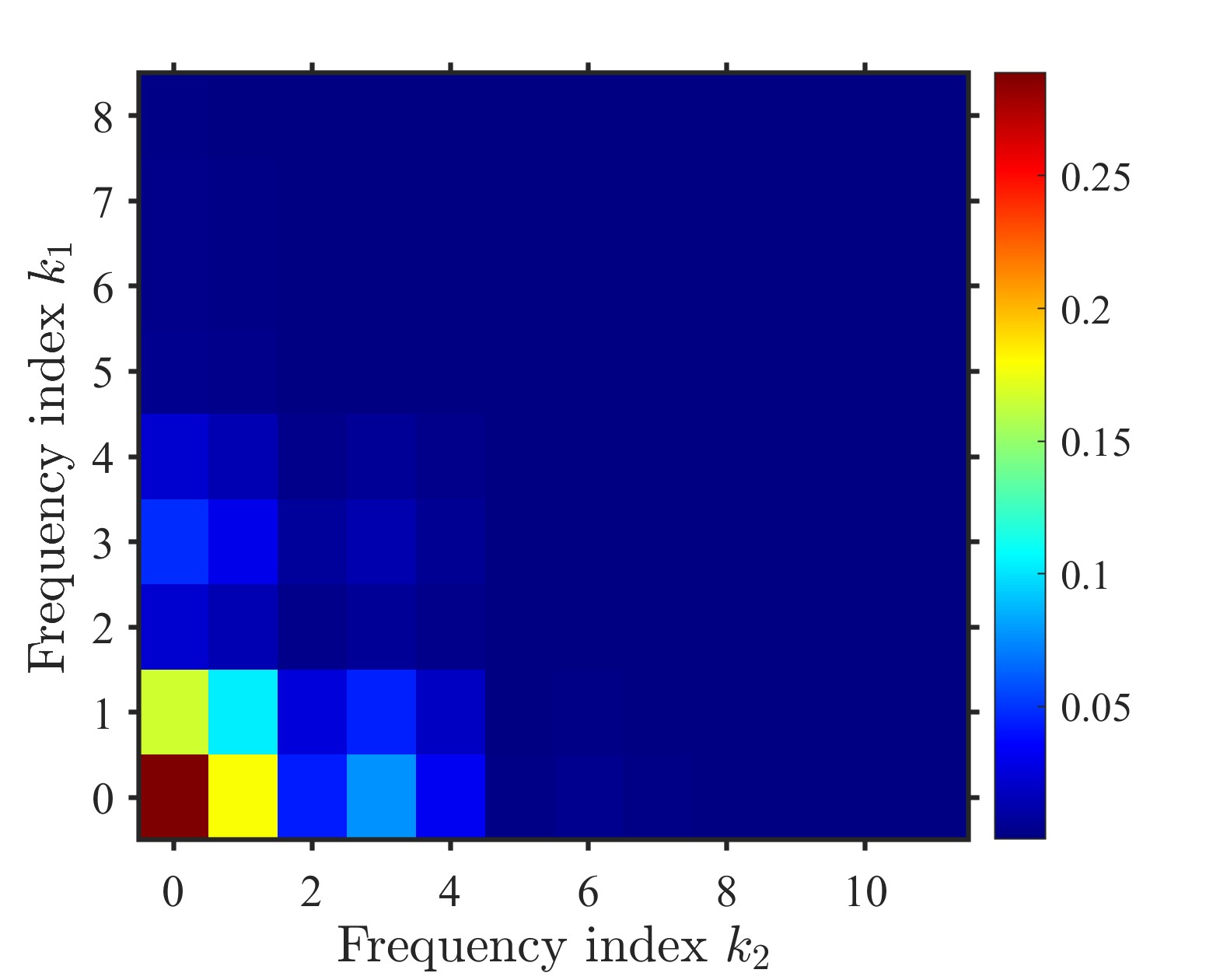}} \hfill
	\subfigure[filtered vertex-domain signal ]{\includegraphics[width=0.245\textwidth]{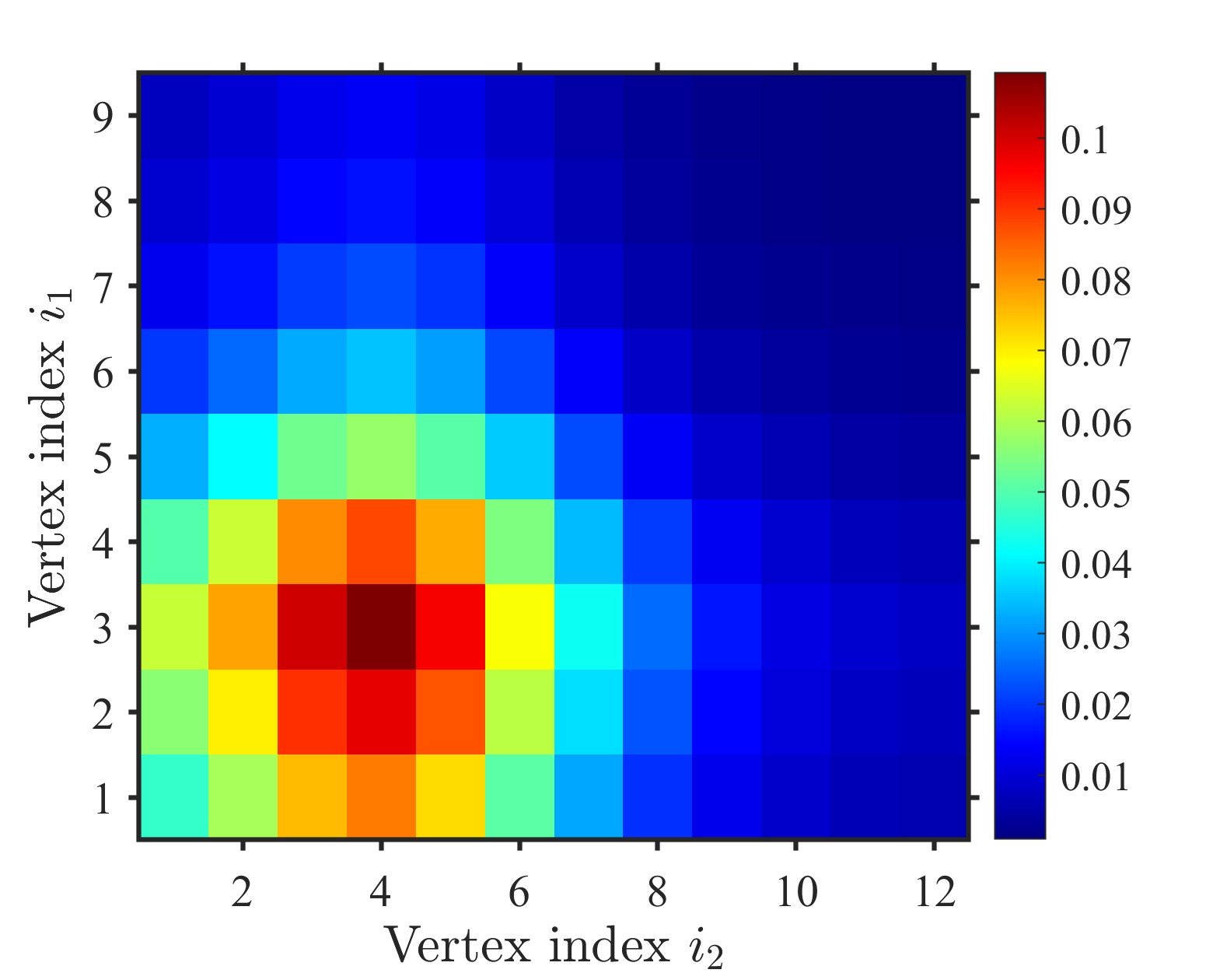}} \hfill
	\caption{Signal and spectral comparison}
	\label{2dfig101}
\end{figure*}

Observations of the frequency-domain spectra indicate that the 2-D transform successfully localizes signal energy in the low-frequency region while effectively suppressing high-frequency stochastic components. However, the most critical distinction arises from the spectral analysis:

\begin{itemize}
	\item \textbf{1-D Spectral Ambiguity:} As shown in Fig. \ref{2dfig101}(b), the spectral coefficients derived from a 1-D projection of the 2-D graph signal exhibit prominent ``multi-valued traits'' at specific graph frequencies. This phenomenon occurs because the 1-D transform maps $N_1 \times N_2$ product eigenvalues onto a single linear frequency axis. Consequently, multiple distinct joint frequency pairs $(r_{k_1}^{(1)}, r_{k_2}^{(2)})$ collapse into identical or near-identical scalar values. This spectral aliasing obfuscates the contribution of individual frequency components to the localized vertex-frequency analysis.
	
	\item \textbf{2-D Spectral Clarity:} Conversely, the 2-D spectrum provides a decoupled and disentangled representation of the signal energy distribution across the frequency indices $(k_1, k_2)$. The 2-D SGFRFT framework inherently alleviates this ambiguity by preserving the intrinsic dimensionality of the spectral domain. As depicted in Fig. \ref{2dfig101}(c), the spectral coefficients are highly concentrated in the low-frequency region of the joint domain, confirming that the 2-D transform robustly resolves the issue of multi-valued spectral coefficients.
\end{itemize}

In conclusion, the 2-D transform framework not only preserves the topological structure of the product graph but also facilitates a much more granular vertex-frequency characterization compared to conventional 1-D methodologies. This precision is indispensable for advanced applications such as multi-dimensional graph denoising, source separation, and feature extraction in complex networks.

\subsection{Algorithm evaluation of F2D-MWGFRFT}

This section evaluates the performance of the proposed fast algorithm (F2D-MWGFRFT). The evaluation focuses on two primary metrics: (i) \textit{Effectiveness}, which assesses acceleration performance and computational complexity, and (ii) \textit{Robustness}, which validates analytical precision and spectral invariance across diverse graph topologies.

\subsubsection{Effectiveness}

To quantitatively demonstrate the computational advantages, we compare the execution time of the F2D-MWGFRFT against the direct 2D-MWGFRFT implementation. Experiments are conducted on random ring graphs with a varying total number of vertices $N_{total} = N_1 \times N_2$ to capture asymptotic behavior. We set the fractional order to $\alpha = 0.7$ and utilize $L=4$ windows derived from a standard Gaussian kernel $\hat{g}(\lambda)$.

Table \ref{2dtab02} presents the relationship between $N_{\text{total}}$ and the computation time (in seconds). For the original 2D-MWGFRFT, the time cost increases significantly, consistent with the theoretical complexity of $O(L(N_1N_2)^4)$. In contrast, the F2D-MWGFRFT scales according to $O(L(N_1N_2)^3)$. The substantial performance gap highlights that the F2D-MWGFRFT achieves a drastic reduction in computation time as the graph dimension increases. This efficiency enables real-time vertex-frequency analysis for large-scale product graphs.
\begin{table}[htbp]
	\centering
	\caption{Computation Time Comparison: 2D-MWGFRFT vs. F2D-MWGFRFT} 
	\label{2dtab02}
	\begin{tabular}{@{}cccccccc@{}}
		\toprule
		{$N_{\text{total}}$}  & 32 & 64 & 128 & 256 & 512 & 1024 & 2048\\
		\midrule
		2D-MWGFRFT &  0.0373 & 0.0497  & 0.1494  & 0.9422  &  5.6994 & 52.0204 & 346.5577 \\ 
		F2D-MWGFRFT &  0.0370 & 0.0474  & 0.0494  &  0.0440 & 0.0728  & 0.3754 & 3.0974 \\
		\bottomrule
	\end{tabular}
\end{table}

\subsubsection{Robustness}
To verify the robustness of the fast algorithm, we analyze three distinct graph signals ($f_2, f_3, f_4$) defined over different topologies: a community graph, a random ring graph, and a sensor graph, each with $N_{total} = 150$. The basis vectors are extracted from the corresponding graph fractional Laplacian matrices.

Setting $\alpha = 0.9$ and employing $L=5$ heat diffusion kernels, we compare the resulting spectrograms. As shown in Fig. \ref{2dfig103}, the vertex-frequency representations generated by the F2D-MWGFRFT are virtually identical to those produced by the standard 2D-MWGFRFT. The algorithm successfully identifies characteristic frequency components and their spatial localizations across all topologies. This high degree of congruence confirms that the optimization in the F2D-MWGFRFT does not introduce numerical dissipation or precision loss.

\begin{figure*}[!t]
	\centering
	\subfigure[a community Cartesian product graph signal $f_2$] {\includegraphics[width=0.32\textwidth]{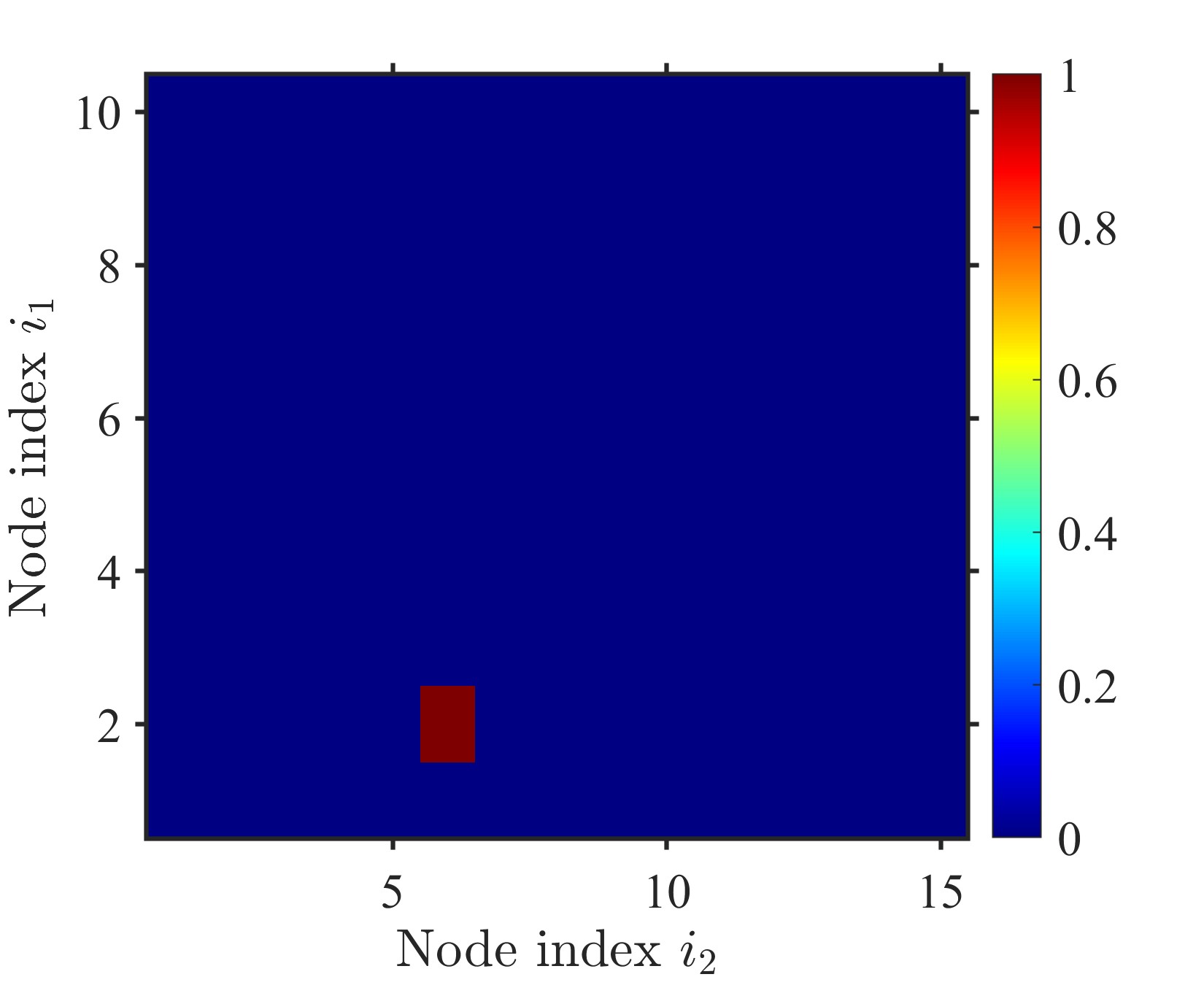}}\hfill
	\subfigure[2D-MWGFRFT of $f_2$]{\includegraphics[width=0.32\textwidth]{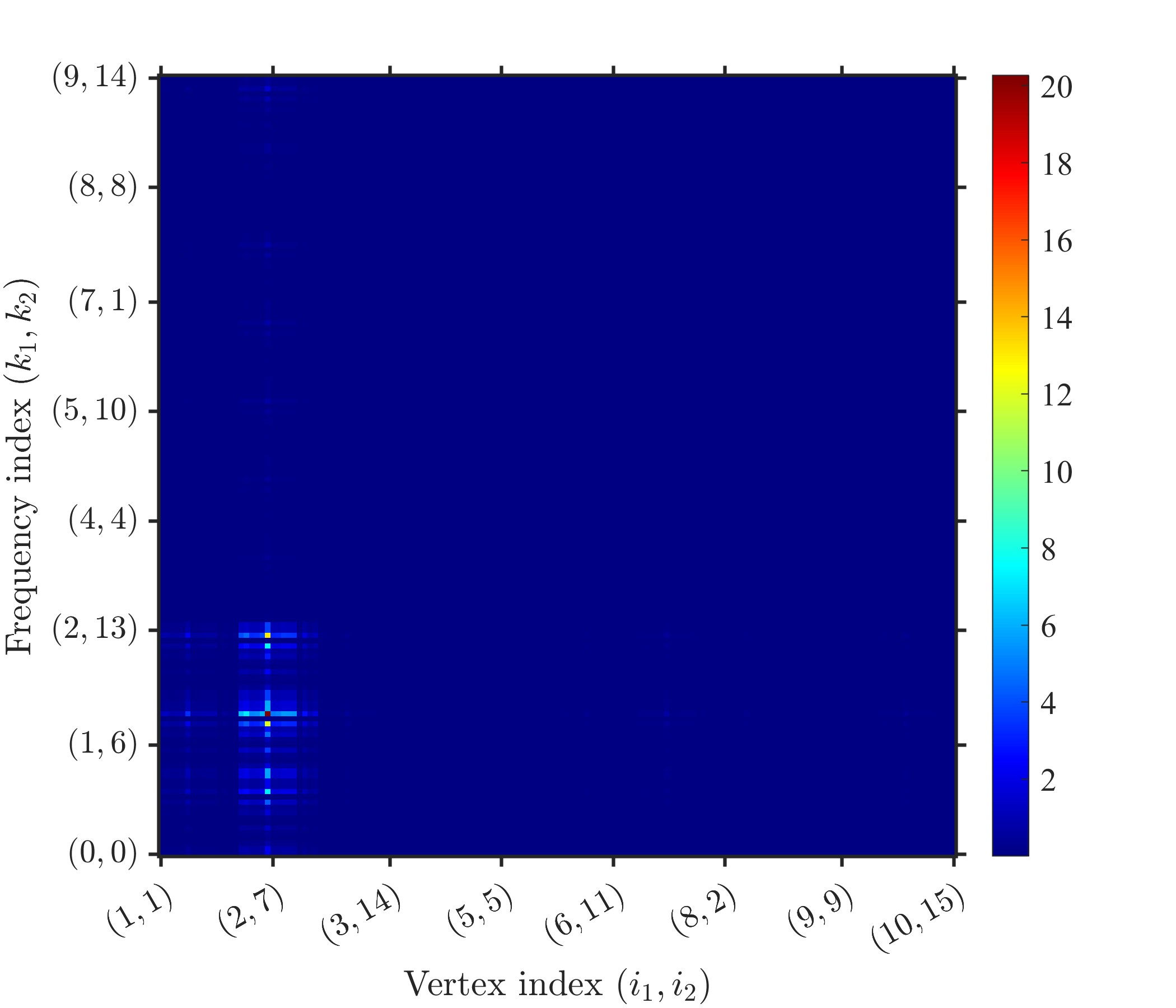}}\hfill
	\subfigure[F2D-MWGFRFT of $f_2$]{\includegraphics[width=0.32\textwidth]{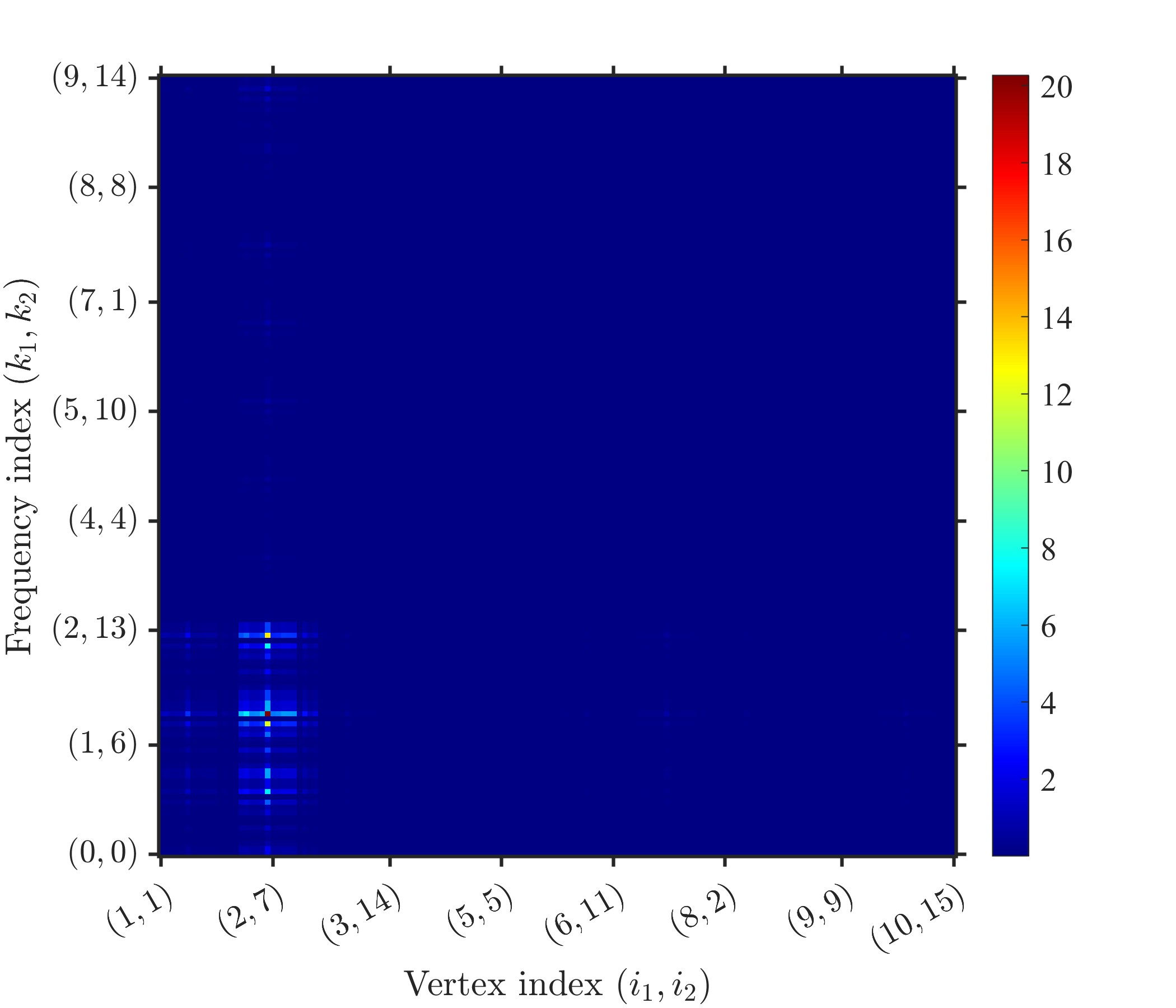}}   \\
	\subfigure[a random ring Cartesian product graph signal $f_3$] {\includegraphics[width=0.32\textwidth]{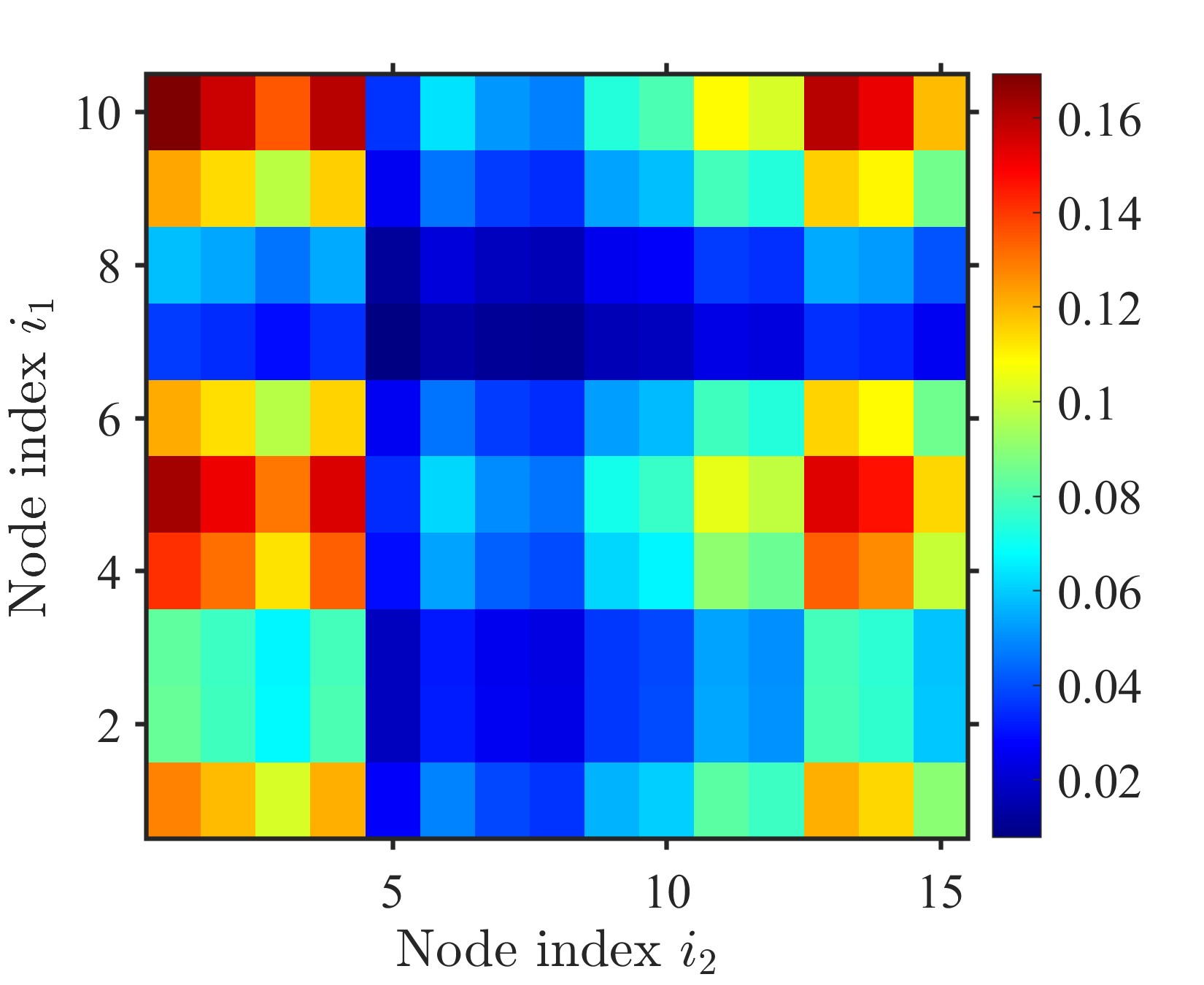}}\hfill     
	\subfigure[2D-MWGFRFT of $f_3$]{\includegraphics[width=0.32\textwidth]{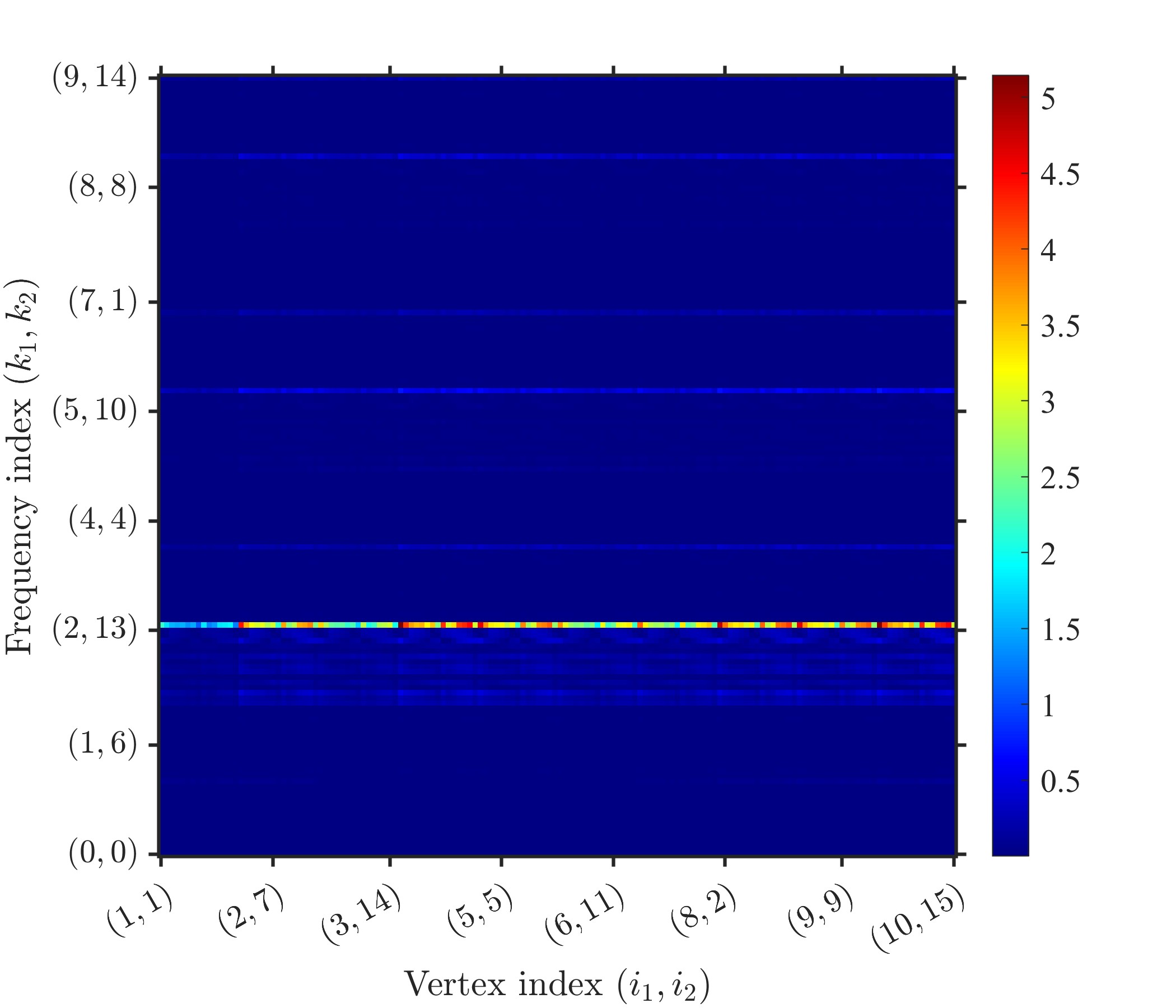}}\hfill
	\subfigure[F2D-MWGFRFT of $f_3$]{\includegraphics[width=0.32\textwidth]{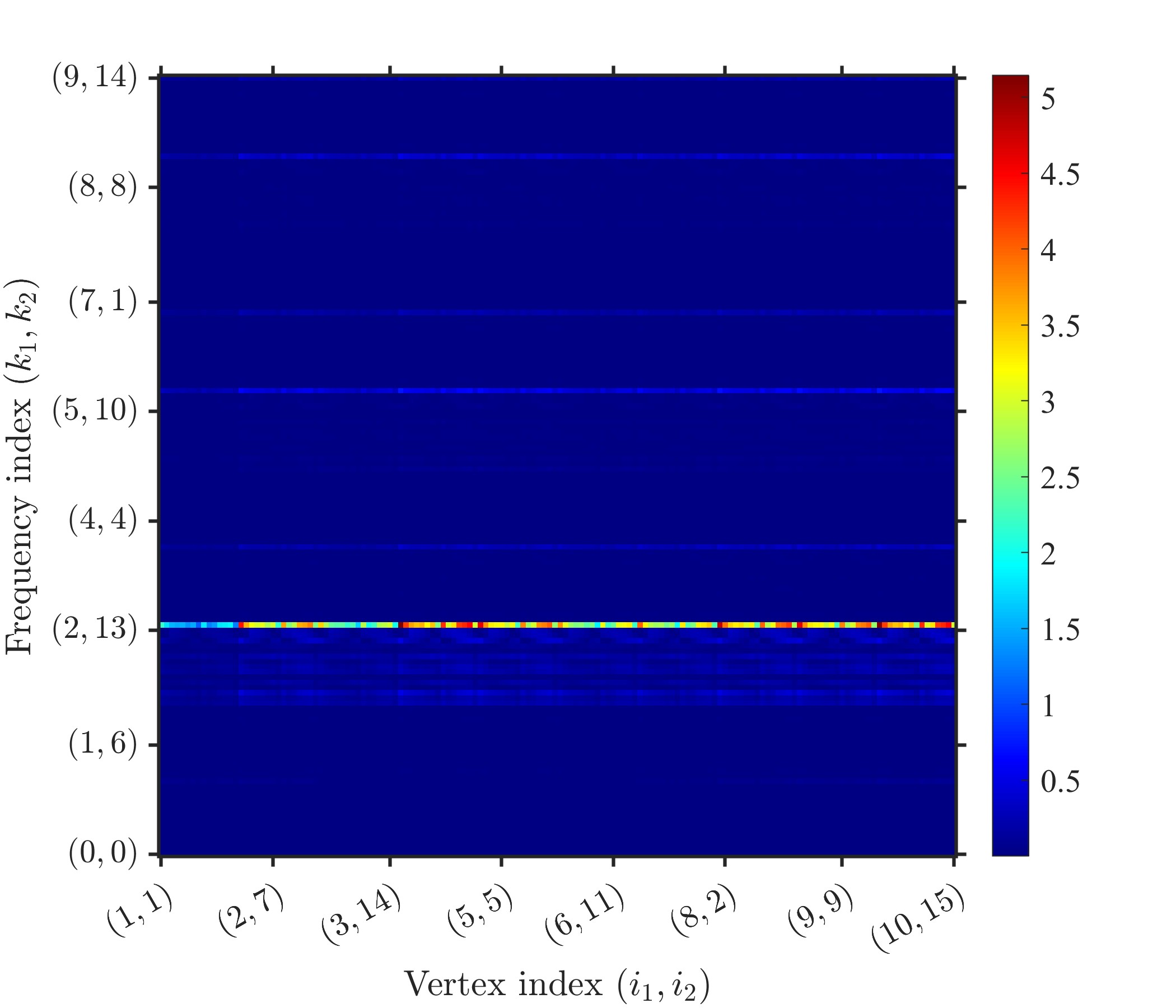}}  \hfill  \\
	\subfigure[a sensor Cartesian product graph signal $f_4$] {\includegraphics[width=0.32\textwidth]{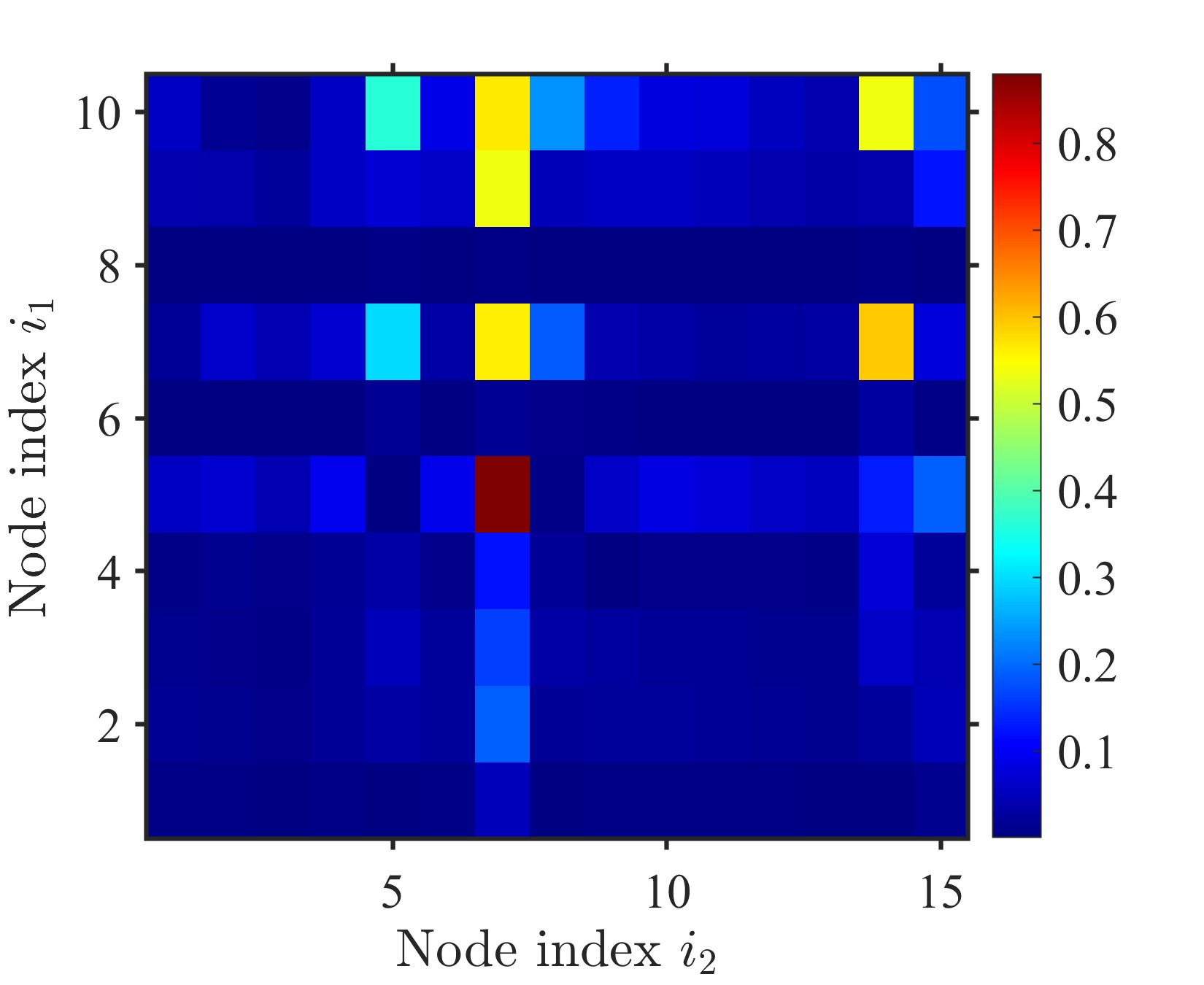}}\hfill
	\subfigure[2D-MWGFRFT of $f_4$]{\includegraphics[width=0.32\textwidth]{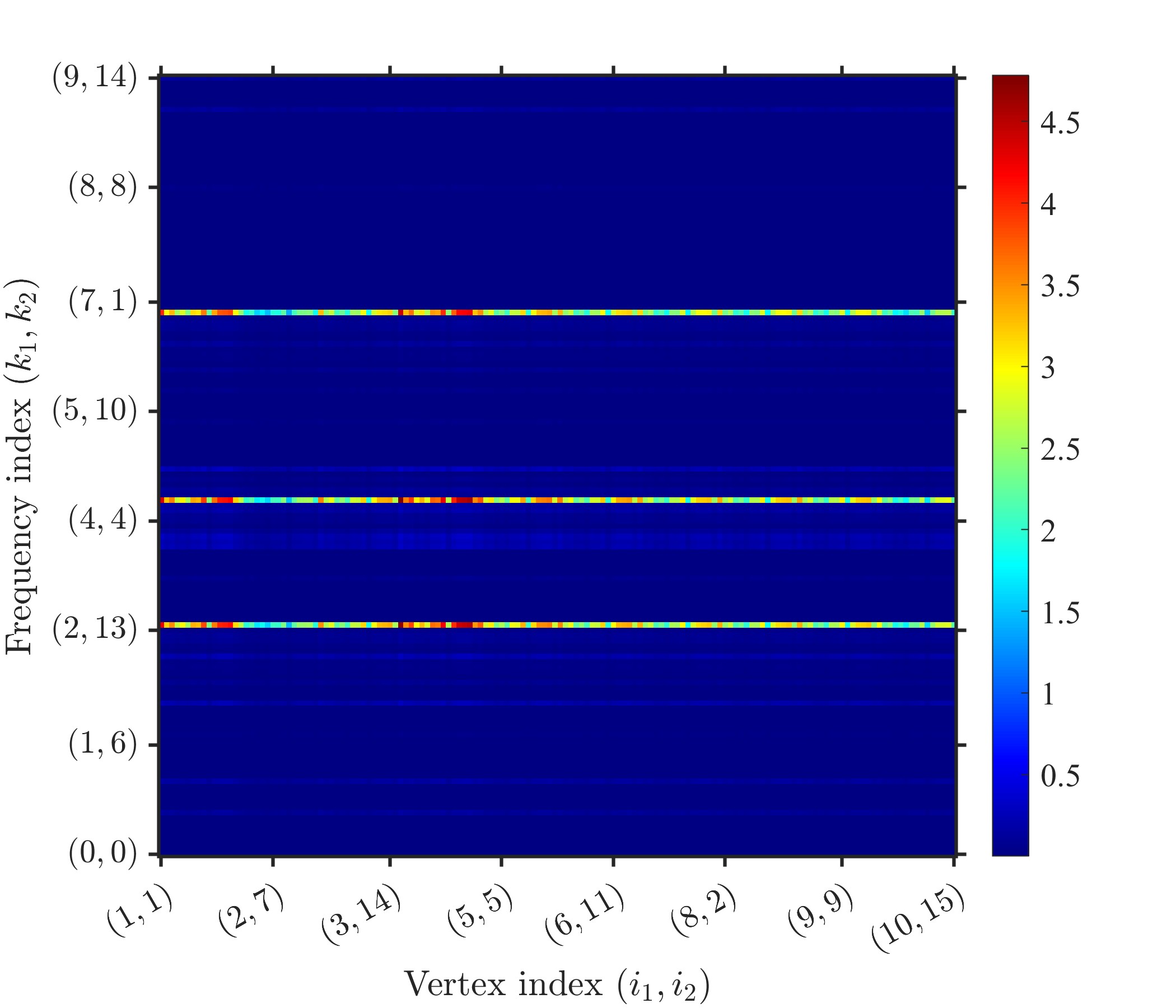}}\hfill
	\subfigure[F2D-MWGFRFT of $f_4$]{\includegraphics[width=0.32\textwidth]{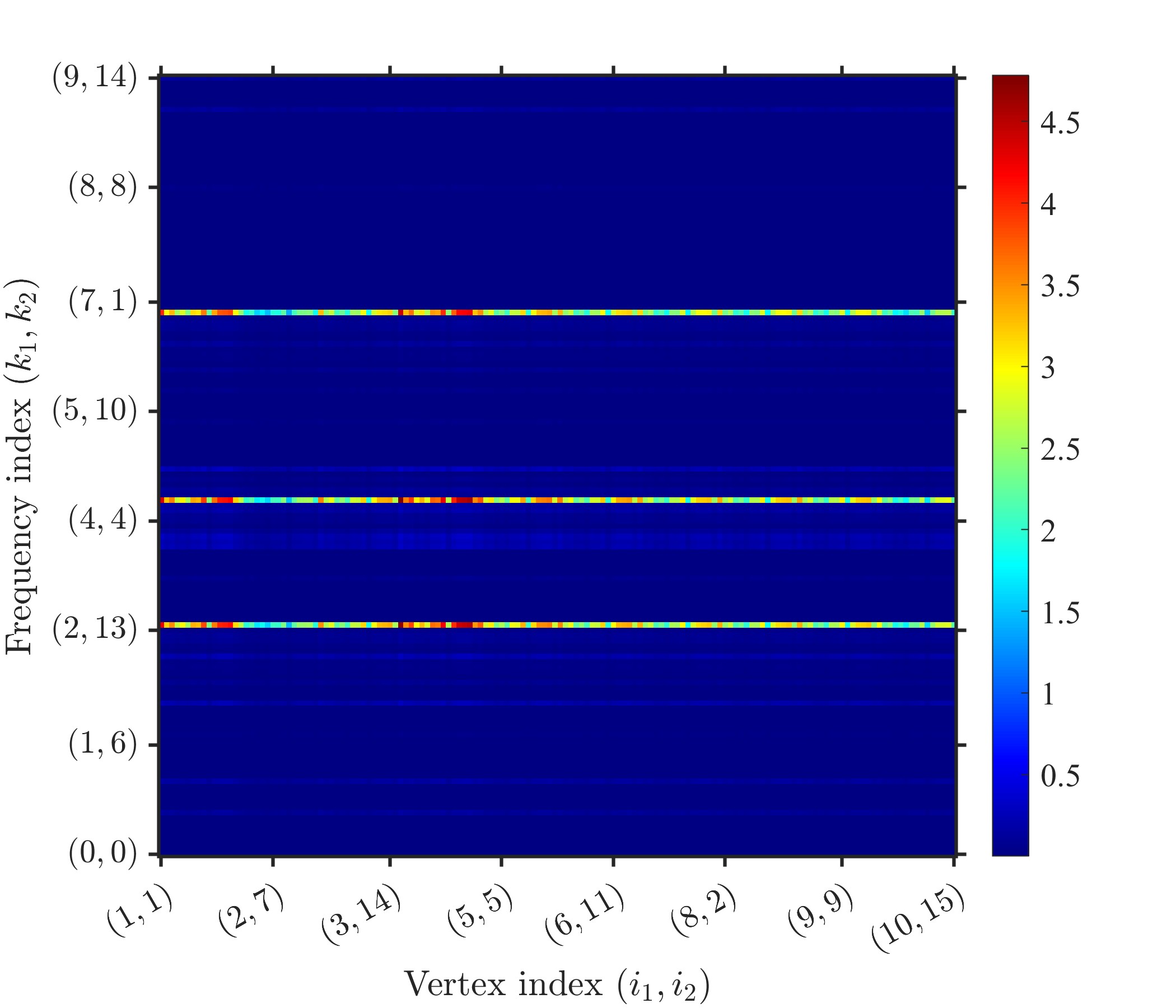}}\\
	\caption{Vertex-frequency representations of signals $f_2$, $f_3$, and $f_4$ by 2D-MWGFRFT and F2D-MWGFRFT, with $\alpha = 0.9$.}
	\label{2dfig103}
\end{figure*}

\subsection{The influence of parameter $L$ on F2D-MWGFRFT}

We investigate the impact of the window count $L$ on the vertex-frequency representation using a product path graph $\mathcal{G} = \mathcal{G}_1 \square \mathcal{G}_2$ ($N_1=12, N_2=10, N=120$). A sparse test signal $f_5$ with an impulse at $f_5(50)=1$ is used to evaluate energy concentration. The window functions $\left\{ {g}_l \right\}_{l=1}^L$ are generated using a Gaussian kernel:
\begin{equation*}
	\widehat{g}_l^{\alpha}(\lambda) = \frac{1}{\sqrt{2\pi}} e^{-\tau_l \lambda^2}, \quad \|\widehat{g}_l^{\alpha}\|_2 = 1
\end{equation*}
Fig. \ref{2dfig104} compares spectrograms for $L \in \{1, 2, 20\}$ at $\alpha = 0.6$.

\begin{figure*}[!t]
	\centering
	\subfigure[a Cartesian product graph signal $f_5$] {\includegraphics[width=0.33\textwidth]{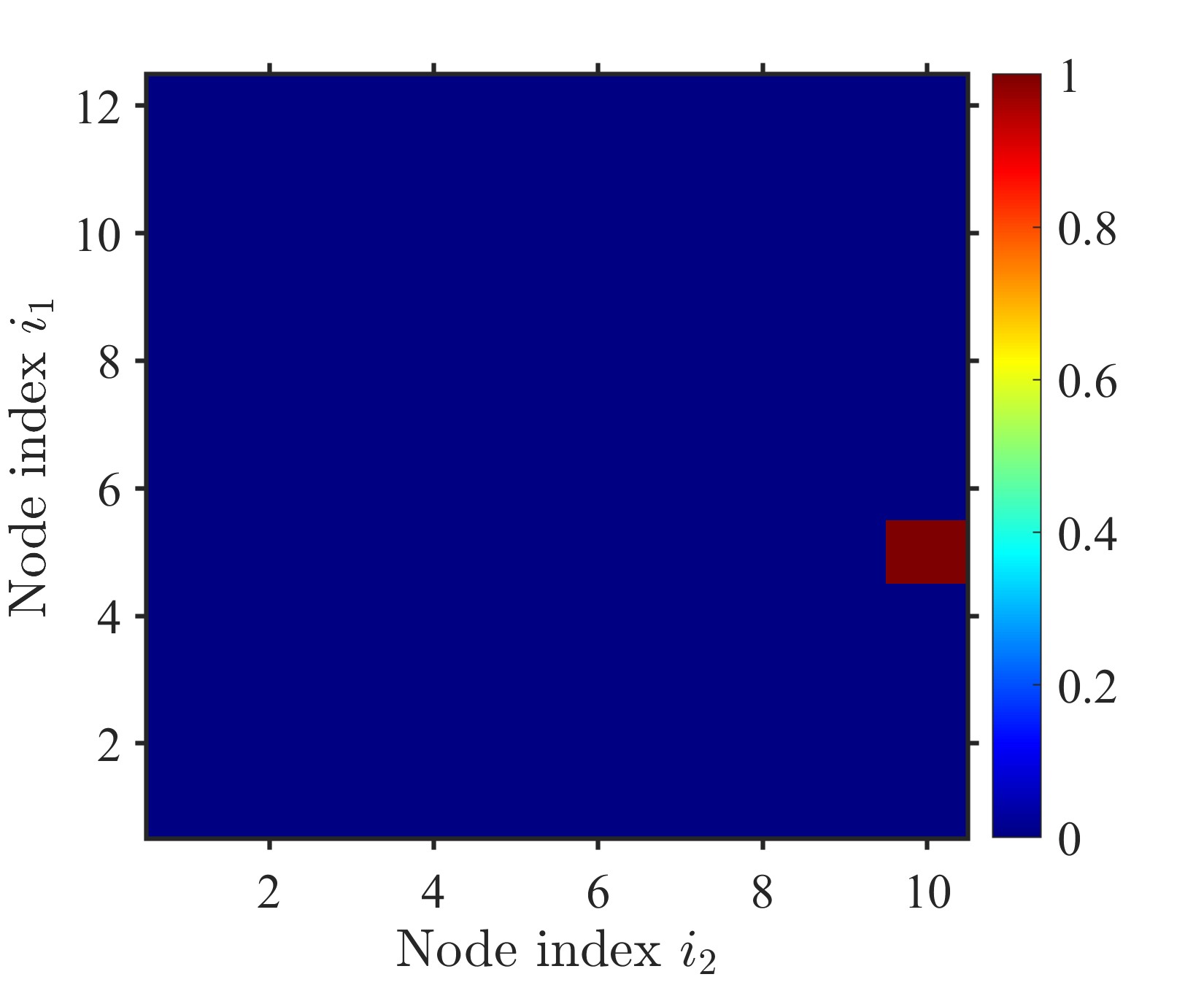}}\hspace{6pt}
	\subfigure[2D-SGFRFT of $f_5$] {\includegraphics[width=0.33\textwidth]{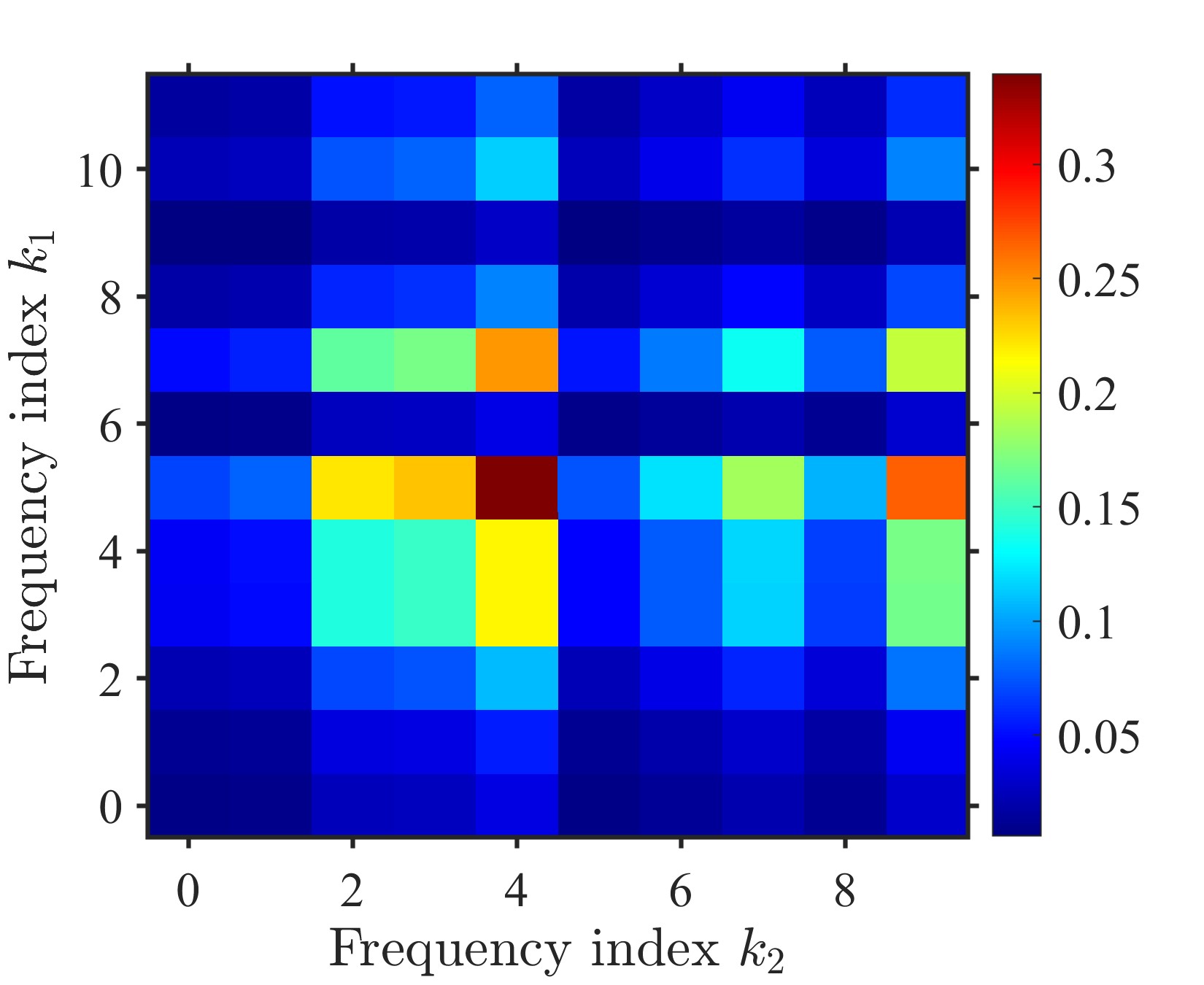}}\hfill \\
	\subfigure[F2D-MWGFRFT of $L=1$] {\includegraphics[width=0.32\textwidth]{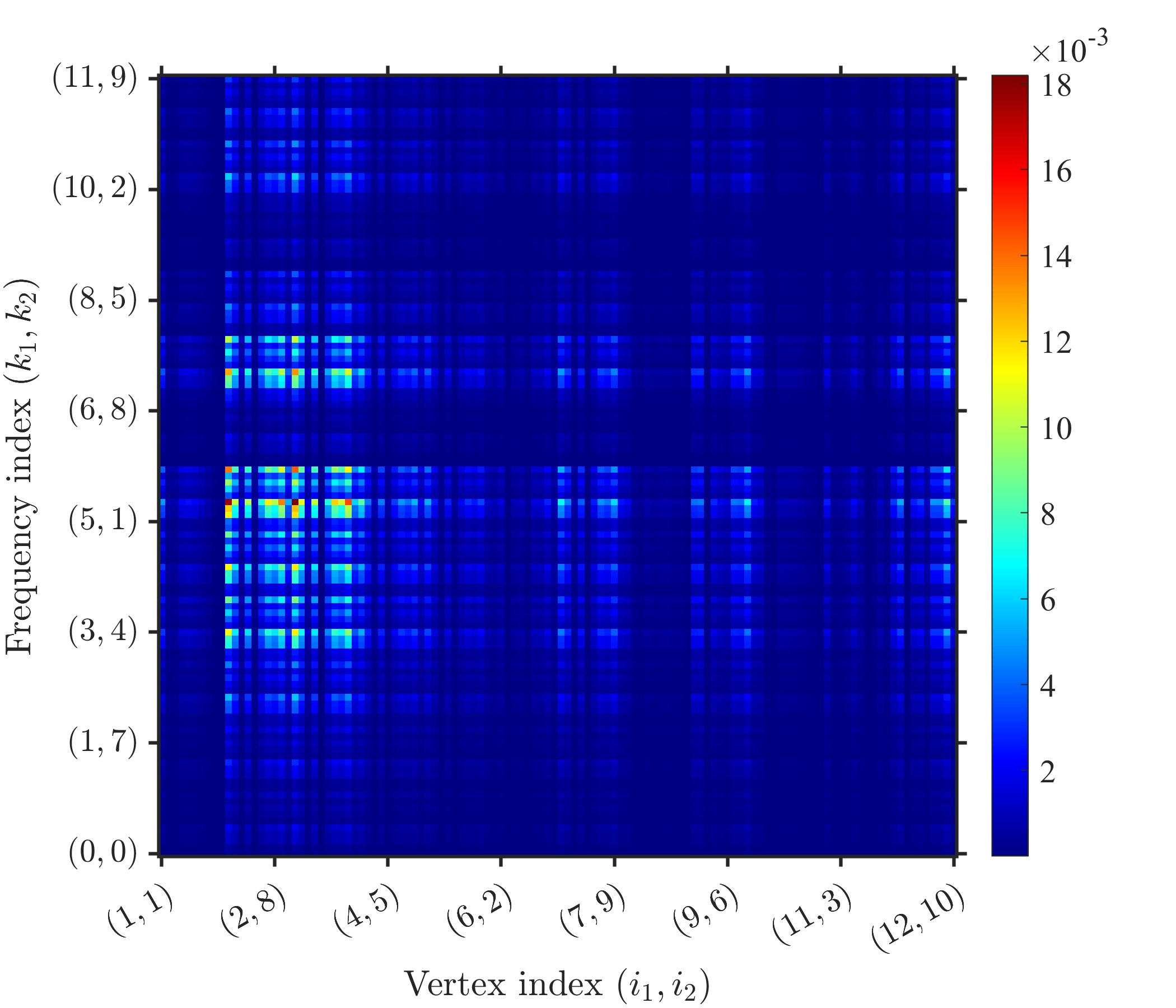}}\hfill     
	\subfigure[F2D-MWGFRFT of $L=2$] {\includegraphics[width=0.32\textwidth]{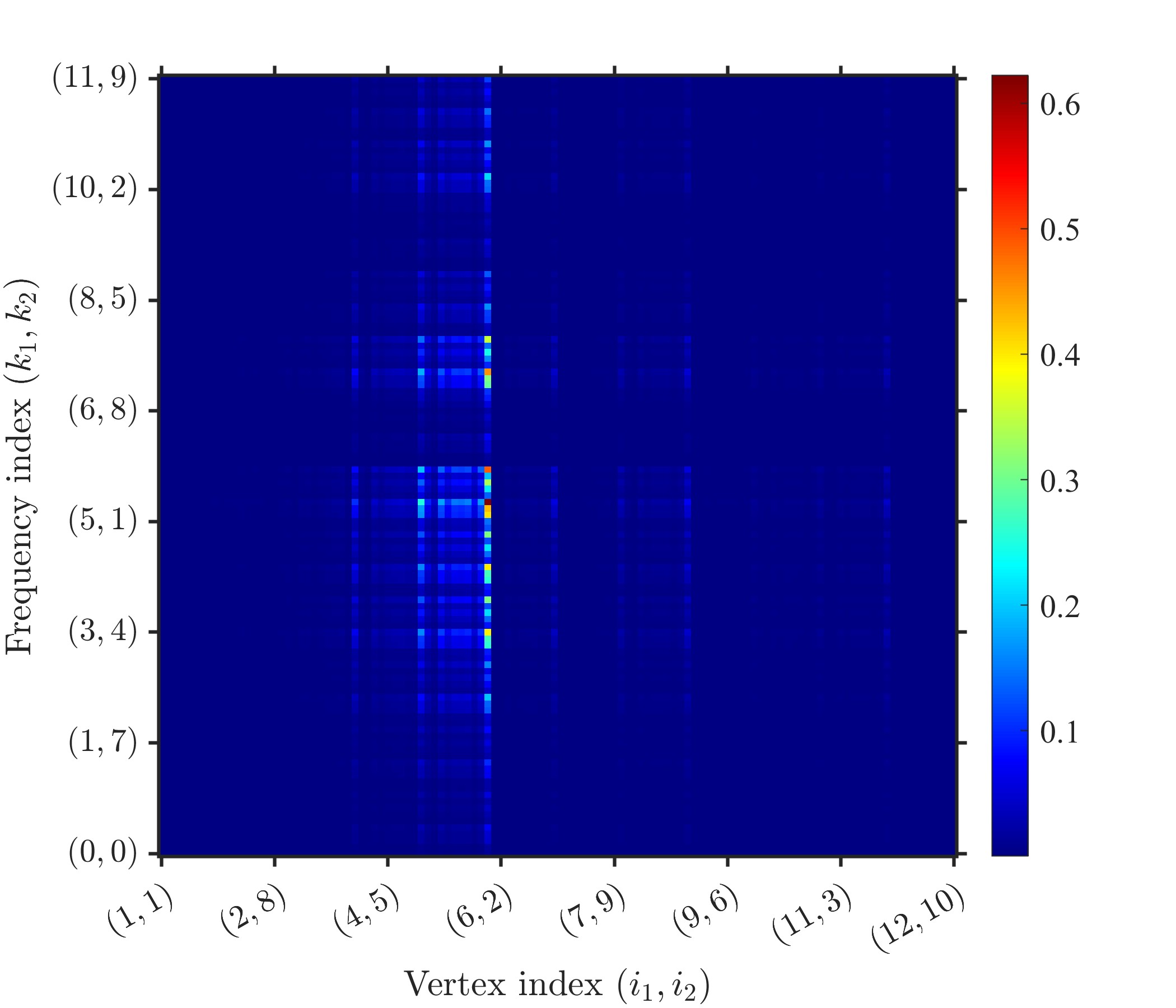}}\hfill
	\subfigure[F2D-MWGFRFT of $L=20$] {\includegraphics[width=0.32\textwidth]{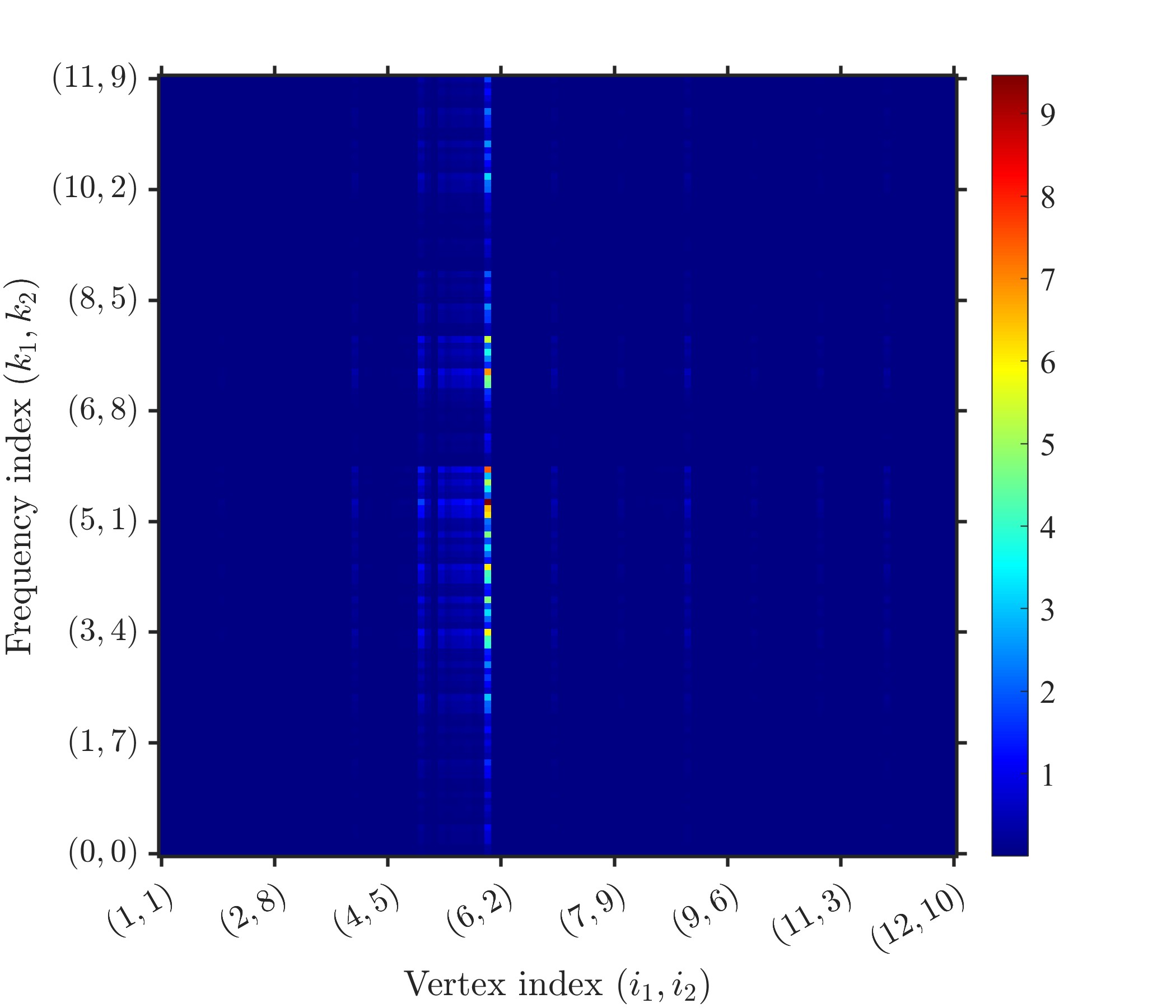}}  \hfill  
	\caption{Vertex-frequency representations of the signal $f_5$ by F2D-MWGFRFT with different $L$.}
	\label{2dfig104}
\end{figure*}

\paragraph{Vertex Localization and Energy Concentration}
The parameter $L$ serves as the primary control for vertex resolution. For $L=1$ (standard 2D-WGFRFT), the energy distribution is relatively dispersed. As $L$ increases to 2, the F2D-MWGFRFT achieves a more compact representation, with energy significantly concentrated around the target vertex index $(i_1,i_2)=(5,10)$. This suggests that the multi-window approach effectively captures local structural features that a single window fails to resolve.

\paragraph{Computational Trade-offs}
Increasing the window count to $L=20$ yields diminishing returns. While the vertex resolution is marginally refined, the gain is less pronounced than the transition from $L=1$ to $L=2$. Since the computational complexity scales linearly with $L$, a high value of $L$ introduces redundant information and overhead without proportional enhancement in representation sparsity.

\paragraph{Summary of Findings}
The experimental results demonstrate that:
\begin{itemize}
	\item \textbf{Low $L$ ($L=1$):} Provides insufficient degrees of freedom, leading to "blurred" vertex-frequency representation and low resolution.
	\item \textbf{Moderate $L$ ($L=2$):} Offers an optimal balance, significantly enhancing vertex localization and yielding a sparse, informative representation of the signal's energy.
	\item \textbf{High $L$ ($L=20$):} Introduces redundant information and increases the computational overhead ($O(L \cdot \text{complexity})$) with marginal improvements in resolution.
\end{itemize}

Thus, for practical applications involving large-scale graph datasets, a moderate value of $L$ (e.g., $L=2$) is recommended to achieve high-fidelity feature extraction while maintaining computational efficiency.

\subsection{Comparative analysis of vertex-frequency feature extraction among 
	F2D-WGFRFT \cite{GAN2025105191} and F2D-MWGFRFT}

To evaluate vertex-frequency localization, we utilize a $12 \times 12$ Cartesian product path graph ($N=144$) with $\alpha = 0.7$. We define signals $f_6$ and $f_7$ with sparse impulses to simulate localized events.

As shown in Fig. \ref{2dfig105}, we define a signal $f_6 \in \mathbb{R}^{N}$ characterized by sparse impulses in the vertex domain to simulate localized "events." Specifically, we set $f_6(5,4) = 1$ and $f_6(7,11) = 1$, representing two distinct point sources on the graph. For the transformation, we employ a heat diffusion kernel as the spectral window function, defined as $\widehat{g}^{\alpha}(\lambda_\ell) = e^{-\tau\lambda_\ell}$, where the energy is normalized such that $\|\widehat{g}^{\alpha}\|_2 = 1$.
As shown in Fig. \ref{2dfig106}, we define a signal $f_7$, representing four distinct point sources on the graph. 
The window functions $\left\{ \mathbf{g}_l \right\}_{l=1}^L$ are generated using a Gaussian kernel in the fractional frequency domain:
\begin{equation*}
	\widehat{g}_l^{\alpha}(\lambda) = \frac{1}{\sqrt{2\pi}} e^{-\tau_l \lambda^2}, \quad \|\widehat{g}_l^{\alpha}\|_2 = 1,
\end{equation*}
where $\tau_l$ denotes the scale parameter.
\begin{figure*}[htbp]
	\centering
	\subfigure[a graph signal $f_6$] {\includegraphics[width=0.245\textwidth]{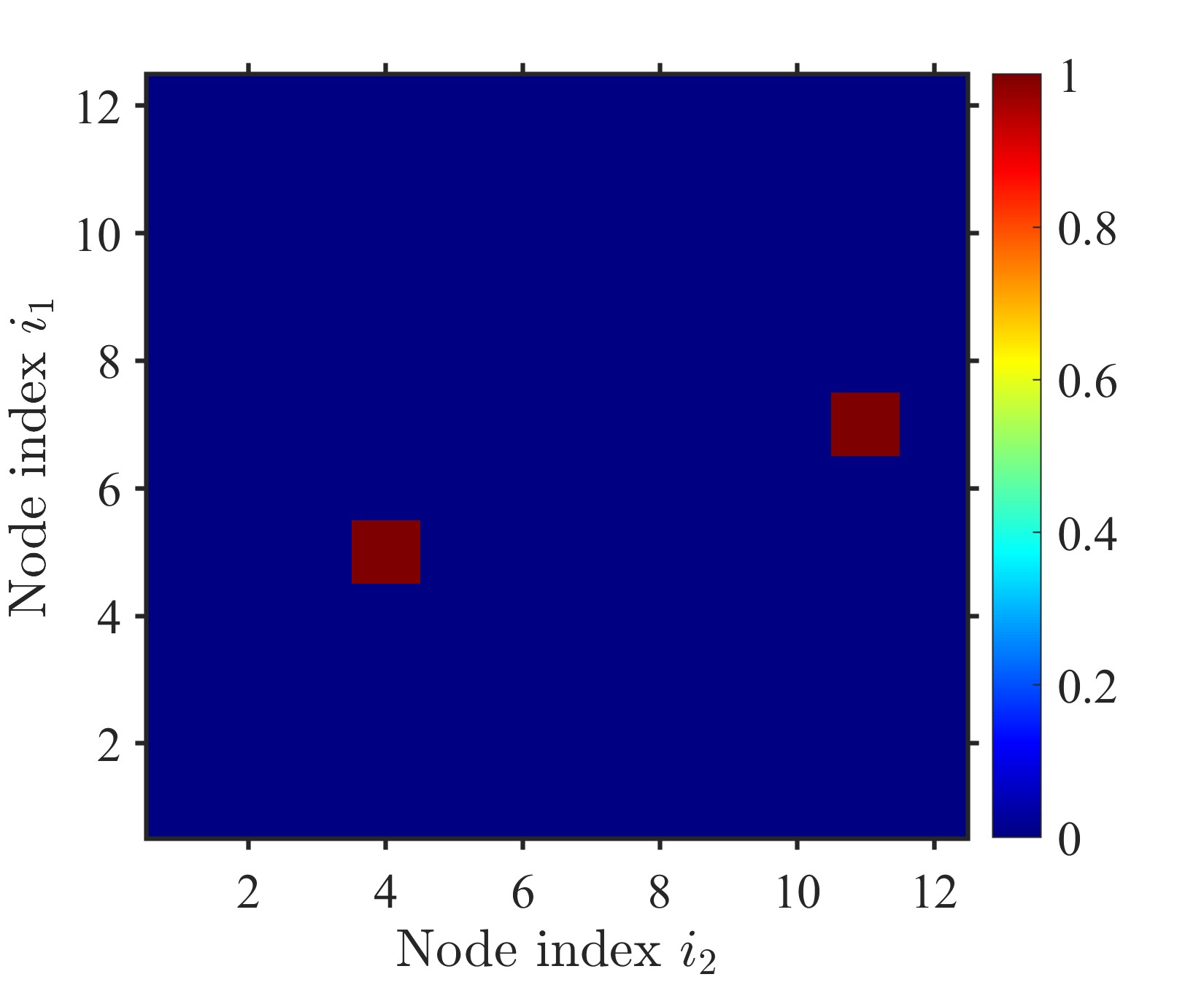}} \hfill
	\subfigure[2D-SGFRFT of $f_6$]
	{\includegraphics[width=0.245\textwidth]{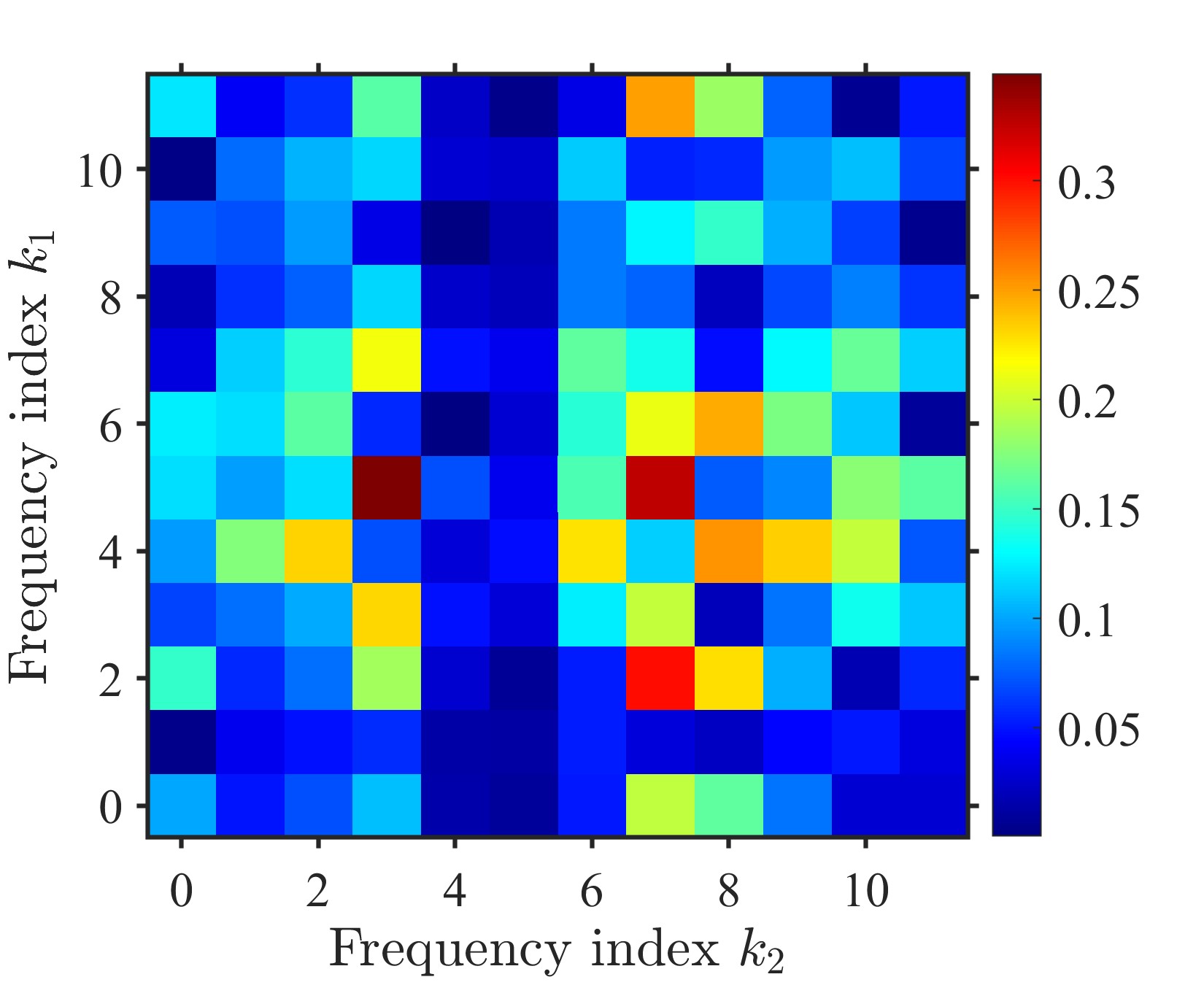}} \hfill
	\subfigure[F2D-WGFRFT of $f_6$
	]{\includegraphics[width=0.245\textwidth]{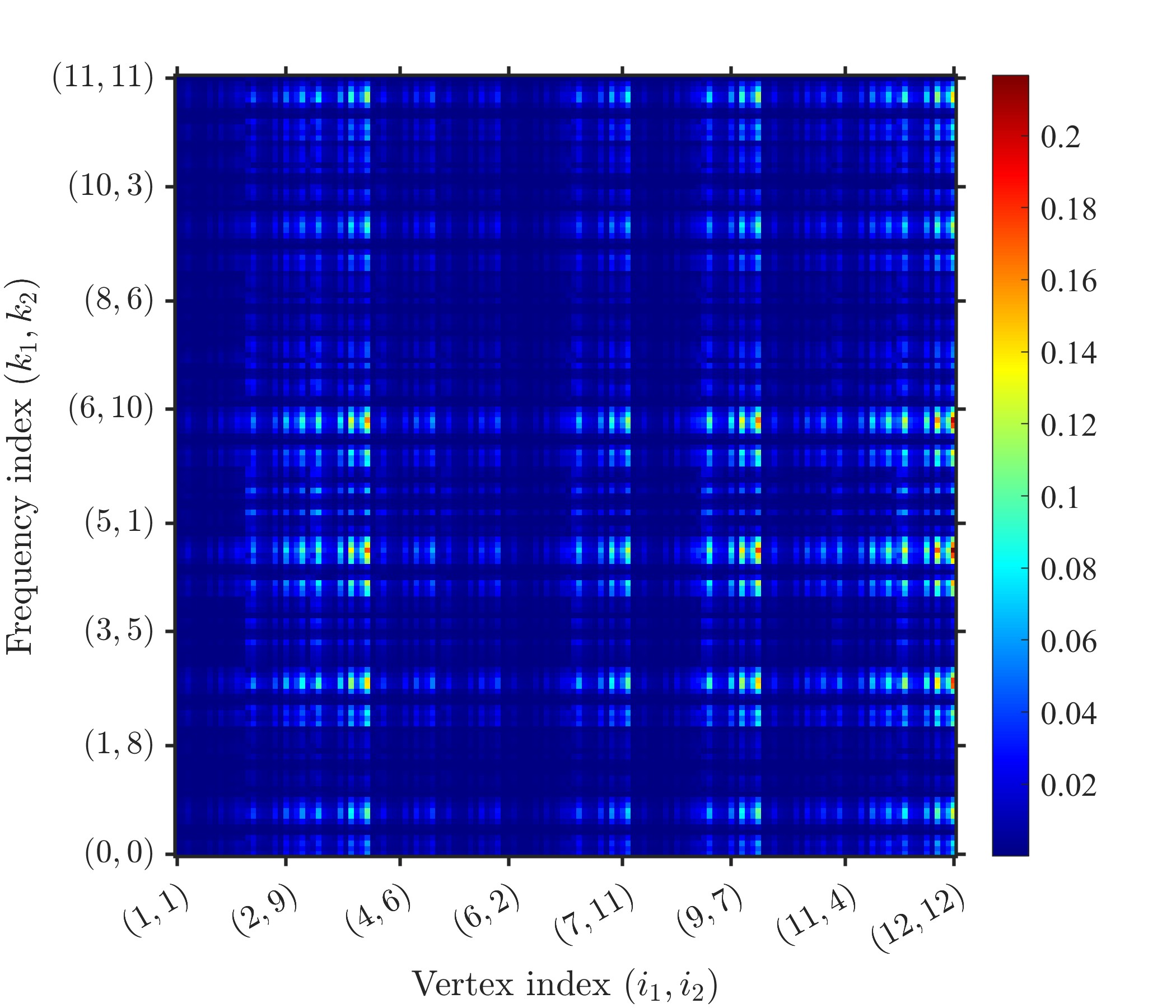}} \hfill
	\subfigure[F2D-MWGFRFT of $f_6$ ]{\includegraphics[width=0.245\textwidth]{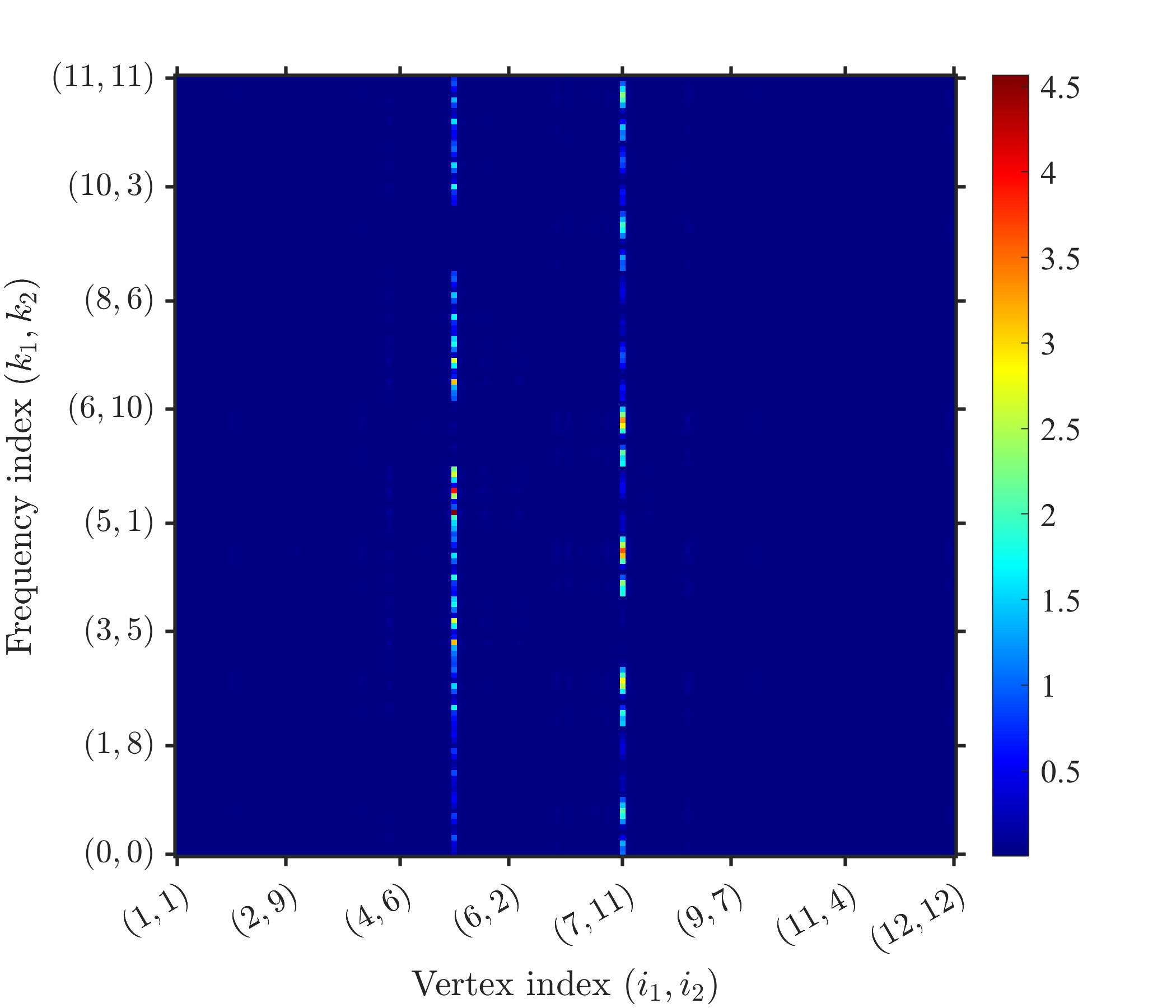}} \hfill
	\caption{Vertex-frequency representations of the signal $f_6$ by F2D-WGFRFT and F2D-MWGFRFT.}
	\label{2dfig105}
\end{figure*}
\begin{figure*}[htbp]
	\centering
	\subfigure[a graph signal $f_7$] {\includegraphics[width=0.245\textwidth]{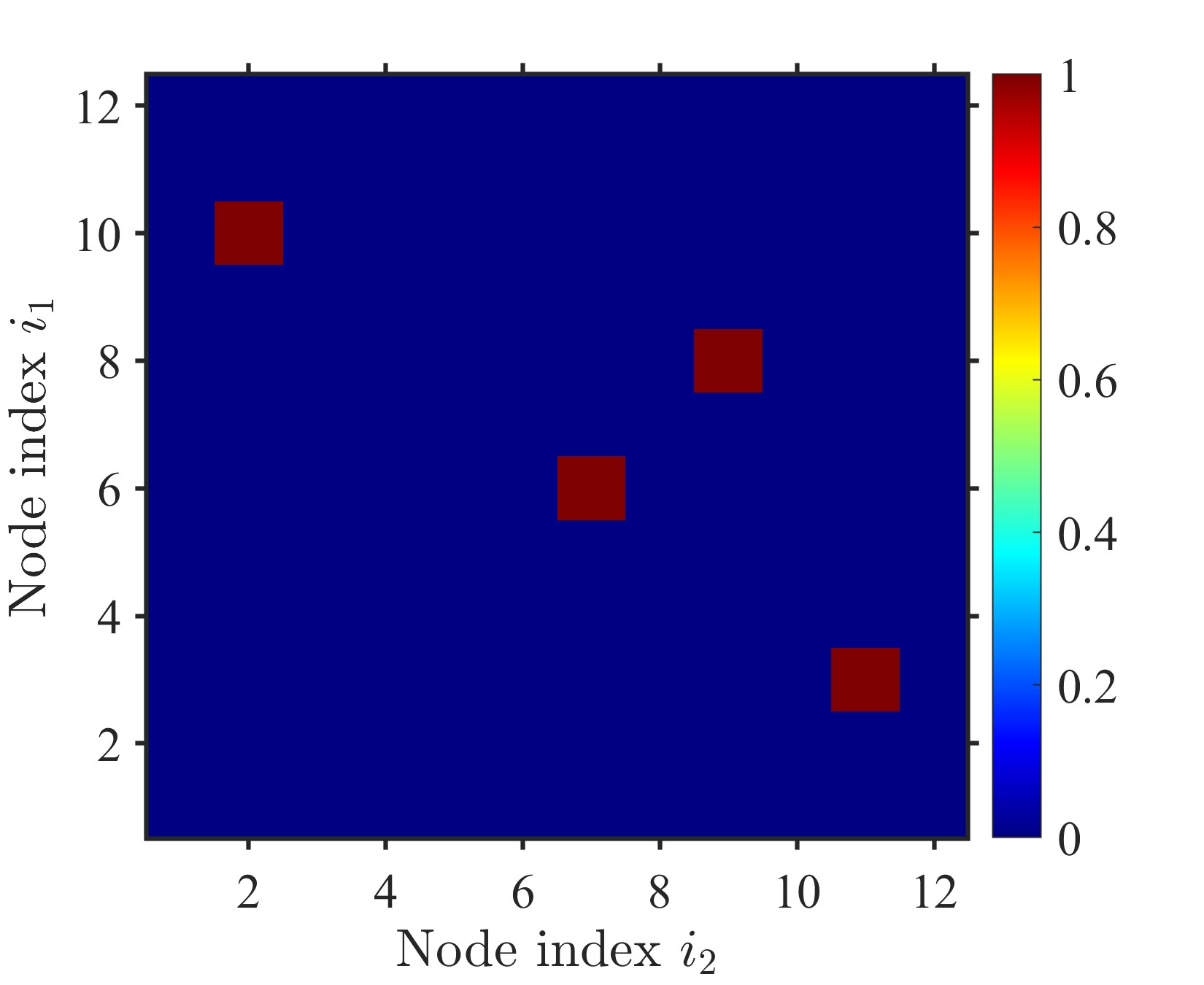}} \hfill
	\subfigure[2D-SGFRFT of $f_7$]
	{\includegraphics[width=0.245\textwidth]{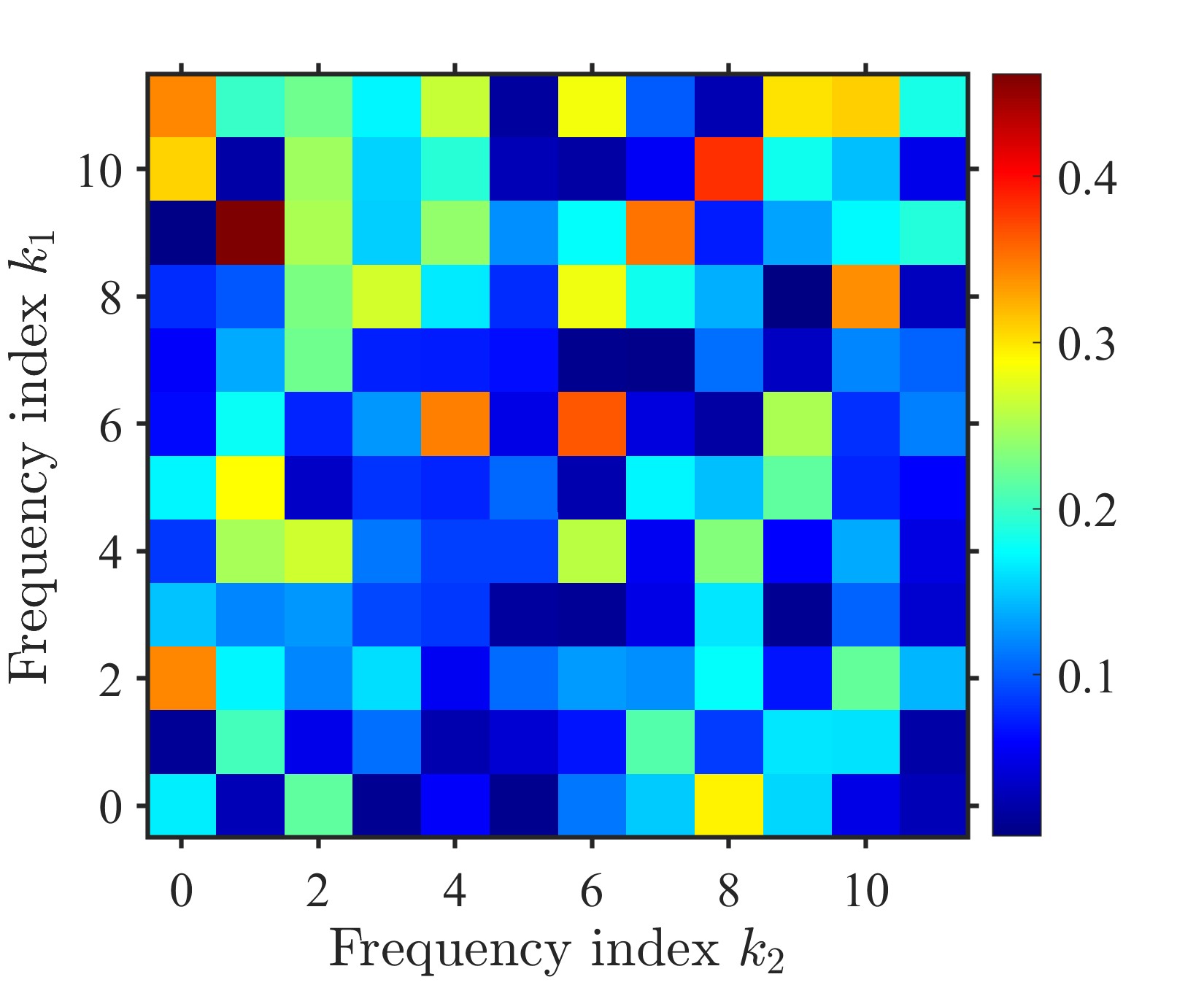}} \hfill
	\subfigure[F2D-WGFRFT of $f_7$
	]{\includegraphics[width=0.245\textwidth]{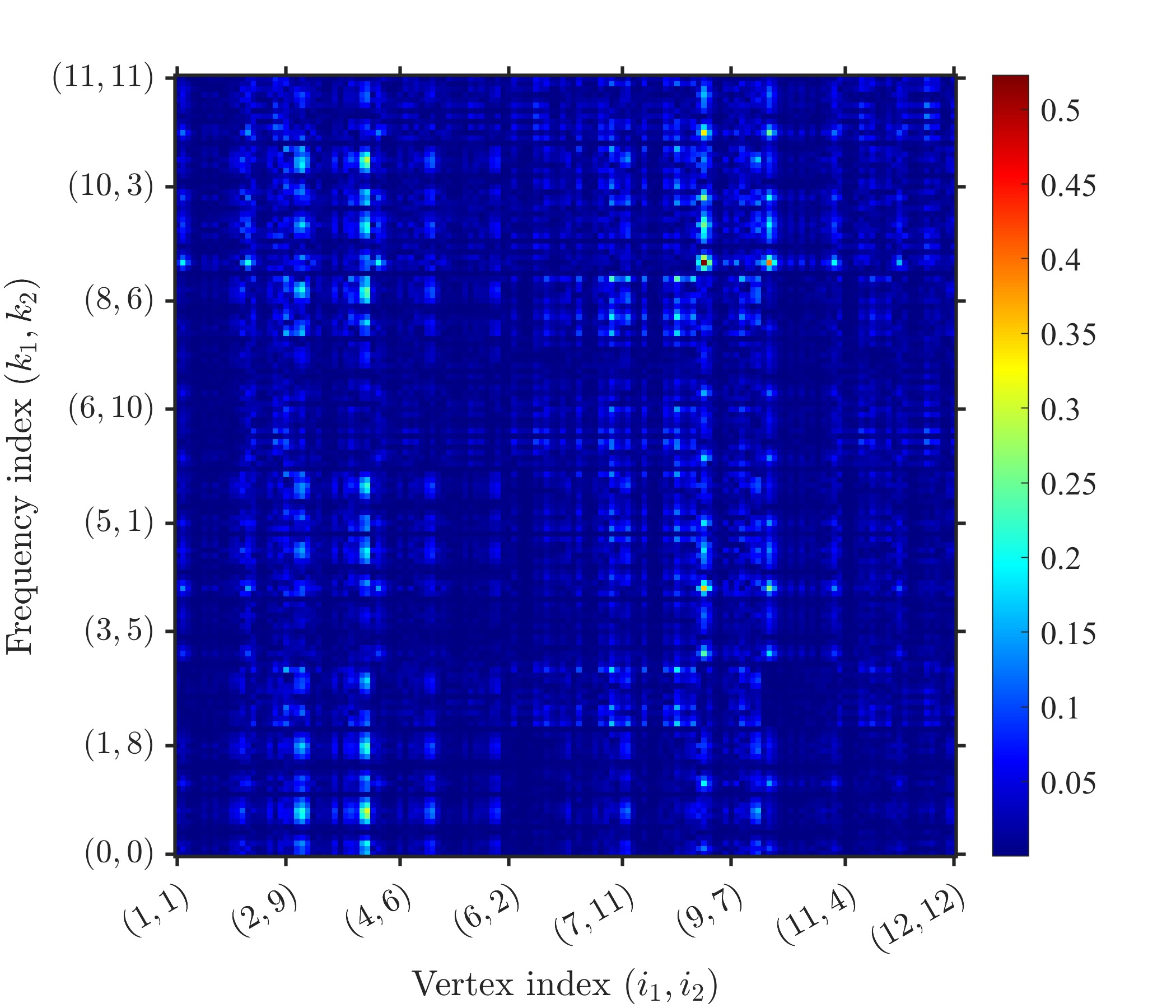}} \hfill
	\subfigure[F2D-MWGFRFT of $f_7$ ]{\includegraphics[width=0.245\textwidth]{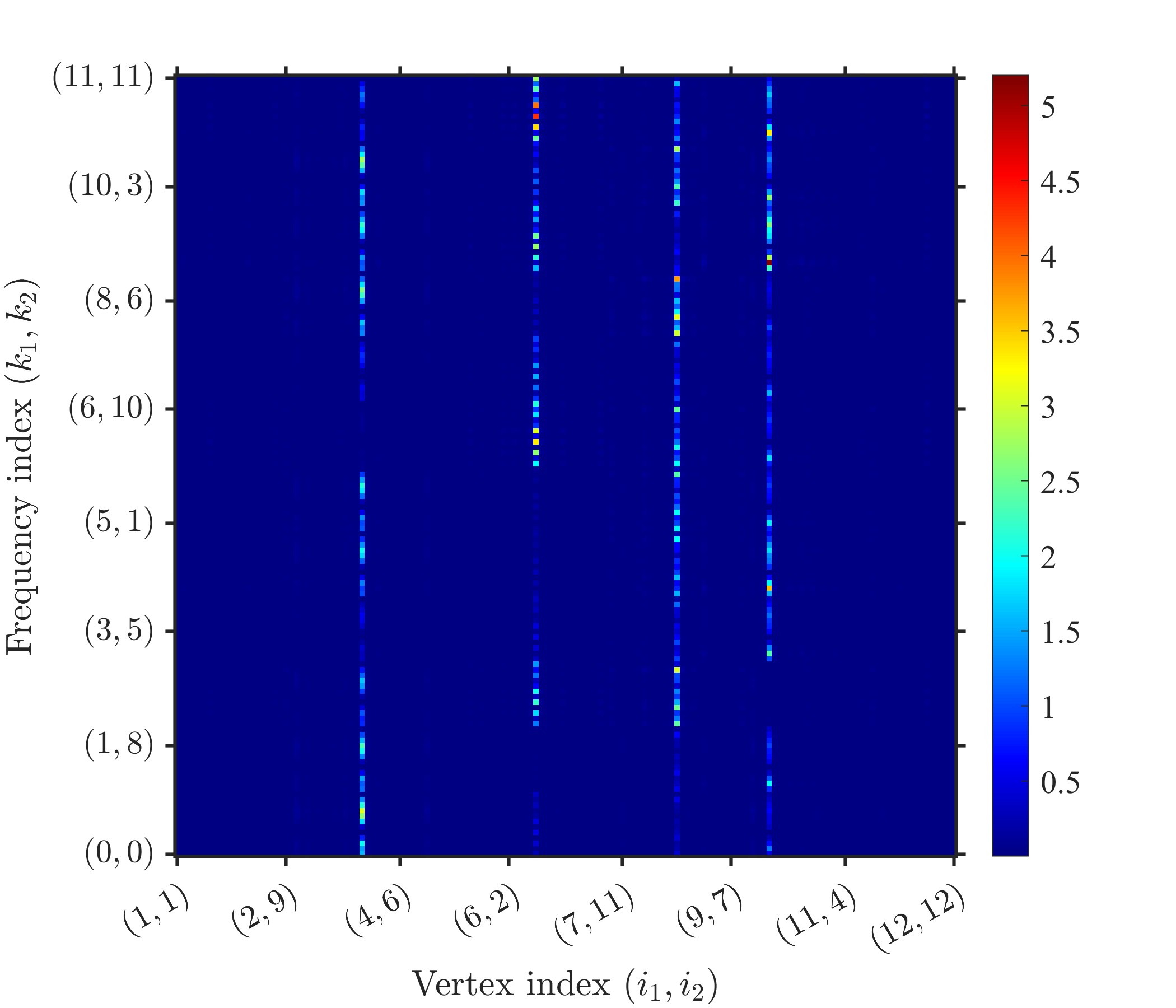}} \hfill
	\caption{Vertex-frequency representations of the signal $f_7$ by F2D-WGFRFT and F2D-MWGFRFT.}
	\label{2dfig106}
\end{figure*}

Fig. \ref{2dfig105} (c) and Fig. \ref{2dfig106} (c) represent the standard F2D-WGFRFT. As shown in the resulting vertex-frequency representations, the single window imposes a fixed resolution trade-off. While the transform identifies the presence of spectral components, the energy remains relatively blurred across the vertex-frequency plane. The single-scale fails to adapt to the local variations of the signal, leading to a loss in precision for both vertex localization and fractional frequency resolution.

By leveraging multiple spectral windows, as shown in Fig. \ref{2dfig105} (d) and Fig. \ref{2dfig106} (d), the F2D-MWGFRFT provides a multiresolution tiling of the vertex-frequency domain. The results indicate that the energy of the signal $f_6$ and $f_7$ is significantly more concentrated. The impulse components are clearly distinguishable as high-intensity "red points" in the spectrogram. Compared to the single-window results, the multiresolution approach accurately captures the sharp transitions in the vertex domain while maintaining high fidelity in the fractional frequency domain.

Ultimately, we conclude that the F2D-MWGFRFT provides a more granular and accurate extraction of vertex-frequency features. The ability to distinguish localized graph signals in the $(\text{vertex}, \text{fractional-frequency})$ plane validates the theoretical advantages of the proposed multiresolution tiling strategy over traditional single-window methods.

\section{Application: anomaly detection on Cartesian product graphs}\label{section5}

To further evaluate the practical performance of the proposed fast two-dimensional multi-window graph fractional Fourier transform (F2D-MWGFRFT), we apply it to anomaly detection on a Cartesian product graph. In this experiment, a Cartesian product graph $\mathcal{G} = \mathcal{G}_1 \square \mathcal{G}_2$ is constructed from two path graphs, each with $N_1=12$ and $N_2=12$ vertices, resulting in a total of $N = 144$ vertices.

We conduct a comparative study between the F2D-WGFRFT \cite{GAN2025105191} and the proposed F2D-MWGFRFT. The spectrogram coefficient matrix $\mathbf{S}$ is computed for both methods. To effectively distinguish anomalies from the background signal, an adaptive threshold $\delta$ is defined based on the maximum coefficient:
\begin{equation*}
	\delta = \frac{1}{2} \max_{i \in \mathcal{V}} \{ \mathbf{S}_i \},
\end{equation*}
where $\mathcal{V}$ denotes the vertex set of the Cartesian product graph. Vertices whose maximum spectrogram coefficients exceed $\delta$ are classified as anomalous.

The experimental results are illustrated in Fig. \ref{2dfig107}. As shown in Fig. \ref{2dfig107}(a), the F2D-WGFRFT spectrogram provides only coarse localization and suffers from limited resolution within the product graph domain. In contrast, the F2D-MWGFRFT (Fig. \ref{2dfig107}(b)) utilizes a bank of five Gaussian windows to facilitate a multiscale analysis, which captures the localized features of the anomalies with significantly higher precision.

\begin{figure*}[htbp]
	\centering
	\subfigure[Spectrogram of F2D-WGFRFT ]{\includegraphics[width=0.32\textwidth]{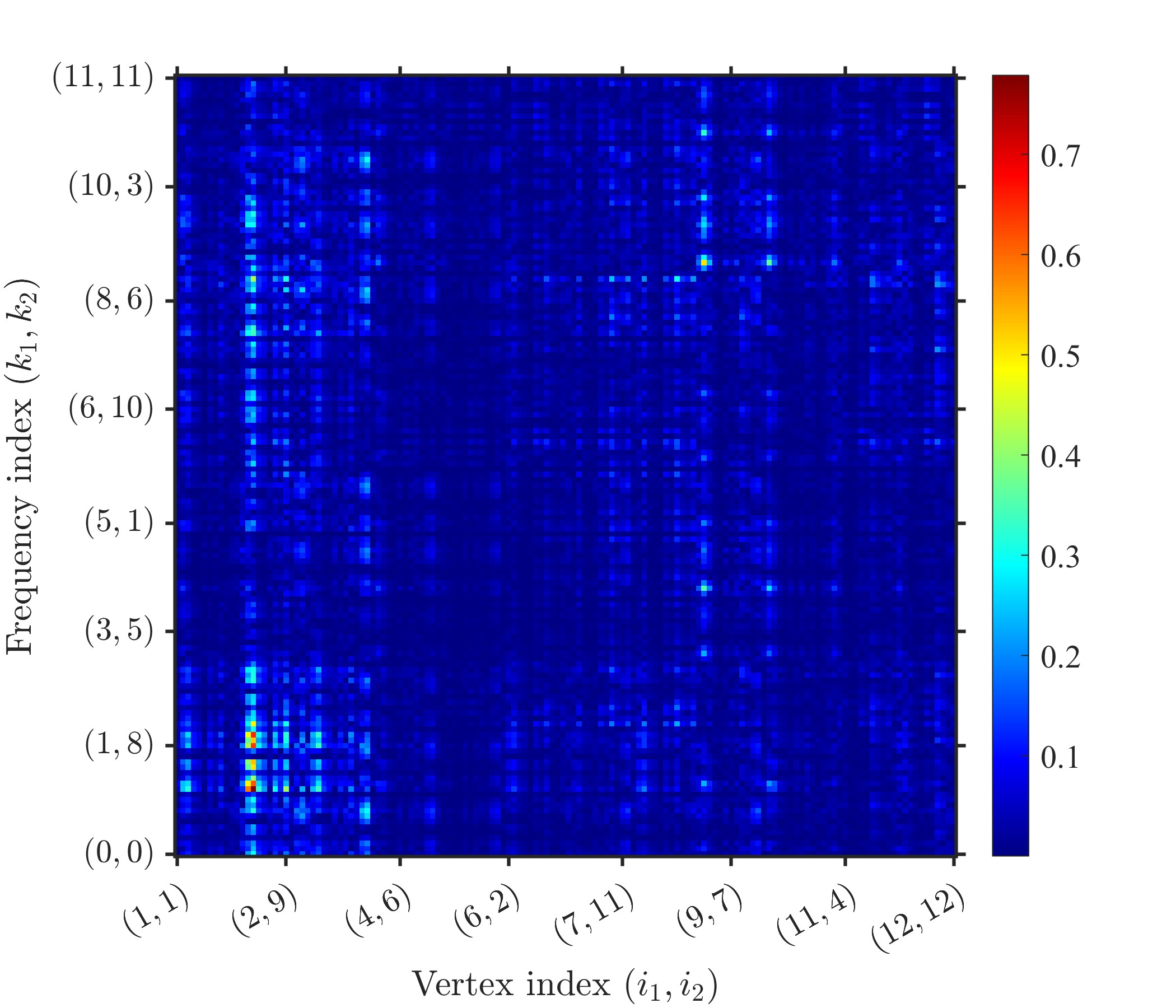}} \hfill
	\subfigure[Spectrogram of F2D-MWGFRFT ]{\includegraphics[width=0.32\textwidth]{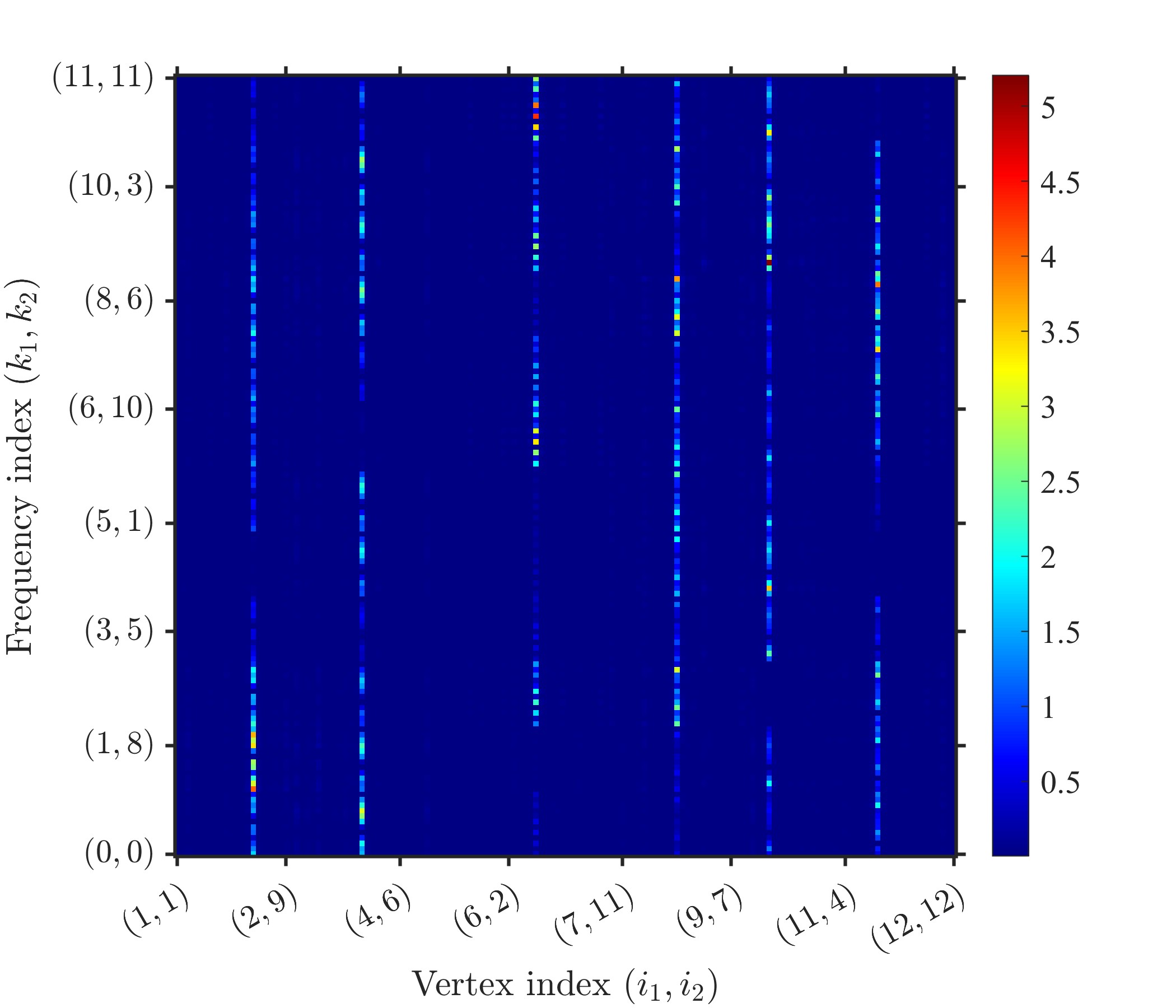}} \hfill
	\subfigure[Detected anomalous vertices of $f_8$] {\includegraphics[width=0.32\textwidth]{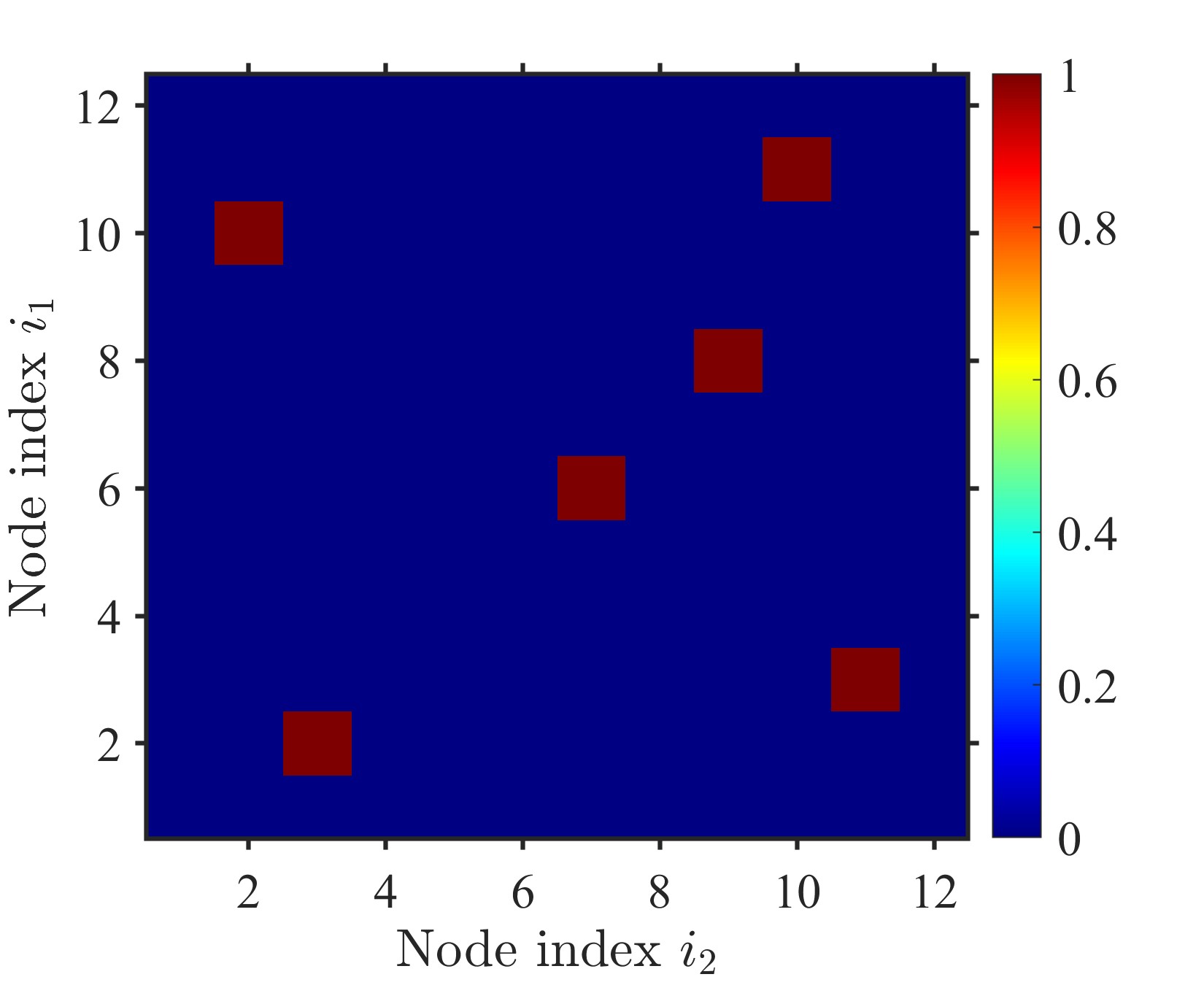}} \hfill
	\caption{Anomaly detection on a $12 \times 12$ Cartesian product path graph: (a) Detection result using a single window (F2D-WGFRFT); (b) Detection result using multiscale windows (F2D-MWGFRFT); (c) Localized anomalous vertices identified by F2D-MWGFRFT with $\alpha = 0.7$.}
	\label{2dfig107}
\end{figure*}

As demonstrated in Fig. \ref{2dfig107}(b), the F2D-MWGFRFT successfully identifies all six anomalous vertices without any false positives. This highlights that by leveraging the multiscale properties of the multi-window framework and the fractional order $\alpha = 0.7$, the proposed method achieves superior energy concentration and detection accuracy compared to single-scale approaches. Specifically, as depicted in Fig. \ref{2dfig107}(c), the F2D-MWGFRFT accurately detects six anomalous vertices at indices $\{(2,3), (3,11), (6,7), (8,9), (10,2), (11,10)\}$. Furthermore, the fast algorithm significantly reduces computational overhead, underscoring its suitability for large-scale product graphs.

\section{Conclusion}\label{section6}

In this paper, we proposed a novel two-dimensional multi-window graph fractional Fourier transform (2D-MWGFRFT) for the analysis of graph signals defined on Cartesian product graphs. By integrating multi-window strategies with fractional spectral operators, the proposed framework enables enhanced vertex-frequency localization and provides a flexible multiresolution representation of graph signals.

We first established the theoretical foundation of the proposed transform by constructing the two-dimensional multi-window graph fractional Fourier frame and deriving the corresponding forward and inverse transforms. This framework ensures stable signal representation and perfect reconstruction. 
To address the computational challenges associated with high-dimensional graph transforms, we further developed an efficient fast algorithm (F2D-MWGFRFT). By exploiting the separable structure of Cartesian product graphs and reformulating the computation in the spectral domain, the proposed method significantly reduces computational complexity, making it suitable for large-scale applications.
Extensive experimental results demonstrated that the proposed method outperforms conventional one-dimensional and single-window approaches in terms of spectral clarity, energy concentration, and vertex-frequency resolution. Moreover, the application to anomaly detection verified the practical effectiveness of the proposed framework in identifying localized features on complex graph structures.

Future work will focus on extending the proposed framework to more general graph products, as well as exploring its applications in deep learning, graph-based signal classification, and large-scale network analysis.

\appendix

\section{Proof of Theorem \ref{th2d02}}\label{2datheorem01}

For any $f \in \mathbb{C}^{N_1 \times N_2}$, we have:
\begin{equation}\begin{aligned}\label{eq2dmw04}
		\sum_{l=1}^L 
		\sum_{i_1=1}^{N_1} \sum_{i_2=1}^{N_2}
		\sum_{k_1=0}^{N_1-1} \sum_{k_2=0}^{N_2-1} 
		|\langle f, g_{i_1,i_2;k_1,k_2}^{(l\alpha)} \rangle|^2 
		&=(N_{1}N_{2})^{\alpha} 
		\sum_{l=1}^L 
		\sum_{i_1=1}^{N_1} \sum_{i_2=1}^{N_2}
		\sum_{k_1=0}^{N_1-1} \sum_{k_2=0}^{N_2-1} 
		\left|\langle f \circ(T_{i_1,i_2}^{\alpha}g_{l})^{*},
		\gamma_{k_1}\gamma_{k_2}\rangle\right|^{2}\\
		&=(N_{1}N_{2})^{\alpha} 
		\sum_{l=1}^L 
		\sum_{i_1=1}^{N_1} \sum_{i_2=1}^{N_2}
		\sum_{n_1=1}^{N_1} \sum_{n_2=1}^{N_2} 
		|f(n)|^{2}|(T_{n_1,n_2}^{\alpha}g_{l})(i_1,i_2)|^{2}\\
		&=(N_{1}N_{2})^{\alpha} 
		\sum_{n_1=1}^{N_1} \sum_{n_2=1}^{N_2} 
		|f(n)|^{2}
		\sum_{l=1}^L |(T_{n_1,n_2}^{\alpha}g_{l})|^{2}
		,\end{aligned}\end{equation}
where \eqref{eq2dmw04} is derived via the Parseval relation, the symmetry of $\mathcal{L}^{(\alpha)}$, and the definition of the operator $T_{i_1,i_2}^{\alpha}.$ 
Crucially, if $\sum_{l=1}^L |\widehat{g}_l^\alpha(0,0)|^2 \neq 0$, then:
$$\sum_{l=1}^L \|T_{n_1,n_2}^{\alpha} g_l\|_2^2
=
(N_{1}N_{2})^{\alpha} \sum_{p_1=0}^{N_1-1}\sum_{p_2=0}^{N_2-1}
\sum_{l=1}^{L}
|\widehat{\mathbf{g}}_{l}^{\alpha}(r_{p_1},r_{p_2})|^{2}
|\gamma_{p_1}(n_1)\gamma_{p_2}(n_2)|^{2}
\geq\sum_{l=1}^{L}\left|\widehat{\mathbf{g}}_{l}^{\alpha}(0,0)\right|^{2}>0
.$$
taking the minimum and maximum of $\sum\limits_{l=1}^L \|T_{n_1,n_2}^{\alpha} g_l\|_2^2$, then we have the lower frame bound $A>0$. 
Thus, $\mathcal{G}^{l\alpha}$ is a frame with lower and upper frame bounds defined in \eqref{eq2dmw05} and \eqref{eq2dmw06}.

\section{Proof of Theorem \ref{th2d01}}\label{2datheorem03}

It can be seen from Definition \ref{def2dmw1} that

$$\begin{aligned}
	&
	\sum_{i_{1}=1}^{N_{1}}\sum_{i_{2}=1}^{N_{2}}
	\sum_{k_{1}=0}^{N_{1}-1}\sum_{k_{2}=0}^{N_{2}-1}
	Wf(g_{i_1,i_2;k_1,k_2}^{(l\alpha)})
	g_{i_1,i_2;k_1,k_2}^{(l\alpha)}(n_1,n_2)\\
	=&
	\sum_{l=1}^{L}\sum_{i_{1}=1}^{N_{1}}\sum_{i_{2}=1}^{N_{2}}
	\sum_{k_{1}=0}^{N_{1}-1}\sum_{k_{2}=0}^{N_{2}-1}
	(N_1N_2)^\alpha
	\sum_{n_1=1}^{N_1}\sum_{n_2=1}^{N_2}f(n_1,n_2)
	\overline{\gamma_{k_1}^{(1)}(n_1)\gamma_{k_2}^{(2)}(n_2)}
	\left(\sum_{p_1=0}^{N_1-1}\sum_{p_2=0}^{N_2-1}
	\widehat{\mathbf{g}}_{l}^{\alpha*}\left(r_{p_1}^{(1)}+r_{p_2}^{(2)}\right)
	\gamma_{p_1}^{(1)}(i_1)\gamma_{p_2}^{(2)}(i_2)
	\overline{\gamma_{p_1}^{(1)}(n_1)\gamma_{p_2}^{(2)}(n_2)} \right) \\
	& \times
	(N_1N_2)^\alpha
	\gamma_{k_1}^{(1)}(d_1)
	\gamma_{k_2}^{(2)}(d_2)
	\left(\sum\limits_{q_1=0}^{N_1-1}\sum\limits_{q_2=0}^{N_2-1}
	\widehat{g}_l^{\alpha}\left(r_{q_1}^{(1)}+r_{q_2}^{(2)}\right)
	\overline{\gamma_{q_1}^{(1)}(i_1)\gamma_{q_2}^{(2)}(i_2)}
	\gamma_{q_1}^{(1)}(d_1)\gamma_{q_2}^{(2)}(d_2) \right)\\
	=&
	(N_1N_2)^{2\alpha}
	\sum_{l=1}^{L}
	\sum_{n_1=1}^{N_1}\sum_{n_2=1}^{N_2}f(n_1,n_2)
	\left(\sum_{p_1=0}^{N_1-1}\sum_{p_2=0}^{N_2-1}\sum_{q_1=0}^{N_1-1}\sum_{q_2=0}^{N_2-1}
	\widehat{\mathbf{g}}_{l}^{\alpha*}\left(r_{p_1}^{(1)}+r_{p_2}^{(2)}\right)
	\widehat{\mathbf{g}}_{l}^{\alpha}\left(r_{q_1}^{(1)}+r_{q_2}^{(2)}\right)
	\overline{\gamma_{p_1}^{(1)}(n_1)\gamma_{p_2}^{(2)}(n_2)} 
	\gamma_{q_1}^{(1)}(d_1)\gamma_{q_2}^{(2)}(d_2) \right) \\
	& \times \sum_{i_{1}=1}^{N_{1}}\sum_{i_{2}=1}^{N_{2}}
	\gamma_{p_1}^{(1)}(i_1)\gamma_{p_2}^{(2)}(i_2)
	\overline{\gamma_{q_1}^{(1)}(i_1)\gamma_{q_2}^{(2)}(i_2)}
	\sum_{k_{1}=0}^{N_{1}-1}\sum_{k_{2}=0}^{N_{2}-1} \overline{\gamma_{k_1}^{(1)}(n_1)\gamma_{k_2}^{(2)}(n_2)}
	\gamma_{k_1}^{(1)}(d_1)
	\gamma_{k_2}^{(2)}(d_2)\\
	=&
	(N_1N_2)^{2\alpha}
	\sum_{l=1}^{L}
	\sum_{n_1=1}^{N_1}\sum_{n_2=1}^{N_2}f(n_1,n_2)
	\left(\sum_{p_1=0}^{N_1-1}\sum_{p_2=0}^{N_2-1}\sum_{q_1=0}^{N_1-1}\sum_{q_2=0}^{N_2-1}
	\widehat{\mathbf{g}}_{l}^{\alpha*}\left(r_{p_1}^{(1)}+r_{p_2}^{(2)}\right)
	\widehat{\mathbf{g}}_{l}^{\alpha}\left(r_{q_1}^{(1)}+r_{q_2}^{(2)}\right)
	\overline{\gamma_{p_1}^{(1)}(n_1)\gamma_{p_2}^{(2)}(n_2)} 
	\gamma_{q_1}^{(1)}(d_1)\gamma_{q_2}^{(2)}(d_2) \right) \\
	& \times 
	\delta_{(p_{1},p_{2})(q_{1},q_{2})}\delta_{(n_{1},n_{2})(d_{1},d_{2})}\\
	=&
	(N_1N_2)^{2\alpha}
	f(n_1,n_2)
	\sum_{l=1}^{L}
	\sum_{p_1=0}^{N_1-1}\sum_{p_2=0}^{N_2-1}
	\left|\widehat{\mathbf{g}}_{l}^{\alpha*}\left(r_{p_1}^{(1)}+r_{p_2}^{(2)}\right)\right|^2
	\left|\overline{\gamma_{p_1}^{(1)}(n_1)\gamma_{p_2}^{(2)}(n_2)}\right|^2 
	.\\
\end{aligned}$$
Moreover,according to Definition \ref{def2dmwt2},it can be found that
$$\begin{aligned}
	\sum\limits_{l=1}^{L}\left\|T_{i_{1},i_{2}}^{\alpha}g_l\right\|_{2}^{2}
	=&
	\sum\limits_{l=1}^{L}
	\sum_{n_{1}=1}^{N_{1}}\sum_{n_{2}=1}^{N_{2}}\left(
	(N_1N_2)^{\alpha/2} \sum_{p_1=0}^{N_1-1}\sum_{p_2=0}^{N_2-1}
	\widehat{{g}}_l^{\alpha}(r_{p_1}^{(1)}+r_{p_2}^{(2)}) \overline{\gamma_{p_1}^{(1)}(i_1)\gamma_{p_2}^{(2)}(i_2)}
	\gamma_{p_1}^{(1)}(n_1)\gamma_{p_2}^{(2)}(n_2)
	\right)^2\\
	=&
	(N_1N_2)^{\alpha}\sum\limits_{l=1}^{L}
	\sum_{p_1=0}^{N_1-1}\sum_{p_2=0}^{N_2-1}
	\sum_{p_1^\prime=0}^{N_1-1}\sum_{p_2^\prime=0}^{N_2-1}
	\widehat{{g}}_l^{\alpha}(r_{p_1}^{(1)}+r_{p_2}^{(2)})
	\widehat{{g}}_l^{\alpha}(r_{p_1^\prime}^{(1)}+r_{p_2^\prime}^{(2)}) 
	\overline{\gamma_{p_1}^{(1)}(i_1)\gamma_{p_2}^{(2)}(i_2)}
	\overline{\gamma_{p_1^\prime}^{(1)}(i_1)\gamma_{p_2^\prime}^{(2)}(i_2)}\\
	& \times 
	\sum_{n_{1}=1}^{N_{1}}\sum_{n_{2}=1}^{N_{2}} \gamma_{p_1}^{(1)}(n_1)\gamma_{p_2}^{(2)}(n_2)
	\gamma_{p_1^\prime
	}^{(1)}(n_1)\gamma_{p_2^\prime
	}^{(2)}(n_2)\\
	=&  
	(N_1N_2)^{\alpha}\sum\limits_{l=1}^{L}
	\sum_{p_1=0}^{N_1-1}\sum_{p_2=0}^{N_2-1}
	\left|\widehat{{g}}_{l}^{\alpha}\left(r_{p_1}^{(1)}+r_{p_2}^{(2)}\right)\right|^2
	\left|\overline{\gamma_{p_1}^{(1)}(n_1)\gamma_{p_2}^{(2)}(n_2)}\right|^2 
	.\end{aligned}$$
Thus, it follows that 
$$\begin{aligned}
	\sum_{i_{1}=1}^{N_{1}}\sum_{i_{2}=1}^{N_{2}}
	\sum_{k_{1}=0}^{N_{1}-1}\sum_{k_{2}=0}^{N_{2}-1}
	Wf(g_{i_1,i_2;k_1,k_2}^{(l\alpha)})
	g_{i_1,i_2;k_1,k_2}^{(l\alpha)}(n_1,n_2)
	=	(N_1N_2)^{\alpha}
	f(n_1,n_2)	\sum\limits_{l=1}^{L}\left\|T_{i_{1},i_{2}}^{\alpha}g_l\right\|_{2}^{2}
	.
\end{aligned}$$
Therefore,there is
$$\begin{aligned}
	f(n_1,n_2)=\frac{1}{(N_{1}N_{2})^{\alpha} \sum\limits_{l=1}^{L}\left\|T_{i_{1},i_{2}}^{\alpha}g_l\right\|_{2}^{2}}
	\times\sum_{i_{1}=1}^{N_{1}}\sum_{i_{2}=1}^{N_{2}}
	\sum_{k_{1}=0}^{N_{1}-1}\sum_{k_{2}=0}^{N_{2}-1}
	Wf(g_{i_1,i_2;k_1,k_2}^{(l\alpha)})
	g_{i_1,i_2;k_1,k_2}^{(l\alpha)}(n_1,n_2).
\end{aligned}$$

\bibliographystyle{elsarticle-num}
\bibliography{main}

\end{document}